\documentclass[a4paper,UKenglish,cleveref, autoref]{lipics-v2019}

\usepackage{amsmath,amssymb}
\usepackage{microtype}
\usepackage[small,basic]{complexity}
\usepackage{xspace}
\usepackage{xcolor}
\usepackage{mathtools}
\usepackage{enumerate}
\usepackage{algorithmic}

\usepackage{tikz}
\usetikzlibrary{positioning,arrows,automata}
\usetikzlibrary{decorations.pathreplacing,decorations.pathmorphing,decorations.text}

\newclass{\EXPTIME}{EXPTIME}
\newclass{\PTIME}{PTIME}
\newclass{\NLOGSPACE}{NLOGSPACE}
\newclass{\LOGSPACE}{LOGSPACE}

\newcommand{\sdomi}[2]{\begin{psmallmatrix}#1\\#2\end{psmallmatrix}} % same as above, but inline
\newcommand{\A}{\mathcal{A}}
\newcommand{\B}{\mathcal{B}}
\newcommand{\sep}{\#}
\newcommand{\pad}{\bot}
\newcommand{\PAD}[2]{{\rm PAD}_{#1}(#2)}
\newcommand{\PADD}[1]{{\rm PAD}_{#1}}

\newcommand{\ASSERT}[1]{\textbf{assert}(#1)}

\newcommand{\defn}[1]{\textit{#1}} %for defining a concept

\newcommand{\N}{\ensuremath{\mathbb{N}}}

\newcommand{\problemx}[3]{
\par\noindent\underline{\sc#1}\par\nobreak\vskip.2\baselineskip
\begingroup\clubpenalty10000\widowpenalty10000
\setbox0\hbox{\bf Instance: }\setbox1\hbox{\bf Question: }
\dimen0=\wd0\ifnum\wd1>\dimen0\dimen0=\wd1\fi
\vskip-\parskip\noindent
\hbox to\dimen0{\box0\hfil}\hangindent\dimen0\hangafter1\ignorespaces#2\par
\vskip-\parskip\noindent
\hbox to\dimen0{\box1\hfil}\hangindent\dimen0\hangafter1\ignorespaces#3\par
\endgroup}

\newcommand{\Lang}{\ensuremath{\mathcal{L}}} %Language
\newif\ifdraft\drafttrue
\ifdraft
\newcommand\anthony[1]{{\color{blue}
[#1 - \textbf{Anthony}]}}

\newcommand{\todo}[1]{{\color{red} TODO: #1}}
\renewcommand{\note}[1]{{\color{orange} NOTE: #1}}
\else
\newcommand\anthony[1]{}
\newcommand\rupak[1]{}

\newcommand\todo[1]{}
\renewcommand{\note}[1]{}
\fi

\newcommand{\OMIT}[1]{\iffalse #1 \fi}

\title{Monadic Decomposability of Regular Relations}

\author{Pablo Barcel{\'{o}}}{Department of Computer Science, University of Chile \& IMFD Chile}{pbarcelo@dcc.uchile.cl}{https://orcid.org/0000-0003-2293-2653}{}%TODO mandatory, please use full name; only 1 author per \author macro; first two parameters are mandatory, other parameters can be empty. Please provide at least the name of the affiliation and the country. The full address is optional

\author{Chih-Duo Hong}{Department of Computer Science, University of Oxford, United Kingdom}{chih-duo.hong@st-hughs.ox.ac.uk}{}{}
\author{Xuan-Bach Le}{Department of Computer Science, University of Oxford, United Kingdom}{bachdylan@gmail.com}{}{}
\author{Anthony W. Lin}{Technische Universit\"{a}t Kaiserslautern,
Germany}{anthony.lin@cs.uni-kl.de}{https://orcid.org/0000-0003-4715-5096}{}
\author{Reino Niskanen}{Department of Computer Science, University of Oxford, United Kingdom}{reino.niskanen@cs.ox.ac.uk}{https://orcid.org/0000-0002-2210-1481}{}

\authorrunning{P. Barcel\'o et al.}%TODO mandatory. First: Use abbreviated first/middle names. Second (only in severe cases): Use first author plus 'et al.'

\Copyright{Pablo Barcel\'o, Chih-Duo Hong, Xuan-Bach Le, Anthony W. Lin and Reino Niskanen}%TODO mandatory, please use full first names. LIPIcs license is "CC-BY";  http://creativecommons.org/licenses/by/3.0/

\category{}%optional, e.g. invited paper

\relatedversion{}%optional, e.g. full version hosted on arXiv, HAL, or other respository/website
\supplement{}%optional, e.g. related research data, source code, ... hosted on a repository like zenodo, figshare, GitHub, ...

\funding{Barcel\'o is funded by the Millennium Institute for Foundational
Research on Data (IMFD) and Fondecyt grant 1170109. Le, Lin, and Niskanen are
supported by the European Research Council (ERC) under the 
European Union's Horizon 2020 research and innovation programme (grant 
agreement no 759969).
}%optional, to capture a funding statement, which applies to all authors. Please enter author specific funding statements as fifth argument of the \author macro.
\acknowledgements{We thank Leonid Libkin for the useful discussion.}%optional

\nolinenumbers %uncomment to disable line numbering

\hideLIPIcs  %uncomment to remove references to LIPIcs series (logo, DOI, ...), e.g. when preparing a pre-final version to be uploaded to arXiv or another public repository

\EventEditors{John Q. Open and Joan R. Access}
\EventNoEds{2}
\EventLongTitle{42nd Conference on Very Important Topics (CVIT 2016)}
\EventShortTitle{ICALP 2019}
\EventAcronym{CVIT}
\EventYear{2016}
\EventDate{December 24--27, 2016}
\EventLocation{Little Whinging, United Kingdom}
\EventLogo{}
\SeriesVolume{42}
\ArticleNo{23}
\begin{document}

\maketitle

\begin{abstract}
    Monadic decomposibility --- the ability to determine 
    whether 
    a formula in a given logical theory can be decomposed into a boolean
    combination of monadic formulas --- is a powerful tool for devising a
    decision procedure for a given logical theory.
    %, and recently was shown
    %to be applicable in the domain of satisfiability modulo theories.
    \OMIT{
    Veanes et al. recently showed
    that this concept could have numerous applications in the domain of
    satisfiability modulo theories (SMT), wherein theories are typically
    quantifier-free. 
    }
    In this paper, we revisit a classical decision problem in automata theory:
    given a regular (a.k.a. synchronized rational) relation, determine whether
    it is recognizable, i.e., it has a monadic decomposition (that is,
    a representation as a boolean combination of cartesian products of regular languages).
    Regular relations are expressive formalisms which, using an
    appropriate string encoding, can capture relations definable in Presburger 
    Arithmetic. In fact, their expressive power coincide with relations
    definable in a universal automatic structure; equivalently, those
    definable by finite set interpretations in WS1S (Weak Second Order Theory 
    of One Successor). 
    Determining whether a regular relation admits a recognizable relation was known 
    to be decidable (and in exponential time for binary relations), but its 
    precise complexity still 
    hitherto remains open. 
    %For binary regular relations, the problem can be decided
    %in exponential time, 
    Our main contribution is to fully settle the complexity of this decision 
    problem by developing new techniques employing infinite Ramsey theory. 
    The complexity for DFA (resp.~NFA) representations
    of regular relations is shown to be $\NLOGSPACE$-complete 
    (resp.~$\PSPACE$-complete). 
    %Finally, we also provide a new application of this
    %classical problem to concurrent program verification.

    %Regular relations are directly connected to concepts in WS1S 
\end{abstract}
    \OMIT{
    As for
    a regular language which is generated by an automaton that runs over a word 
    of over a given alphabet, an $n$-ary regular relation is generated by an 
    automaton which runs over a word of a product alphabet in a synchronized
    fashion. 
    }

\section{Introduction}
\label{sec:intro}

Monadic decompositions for computable relations have been studied in many 
different guises, and applied to many different problem domains, e.g., see
\cite{CC79,Lib00,VBN+17,CCG06,LS17,LS19,Valiant75}. 
The notion of ``monadic decomposability''
essentially captures the intuitive notion that the components in
a given $n$-ary relation $R \subseteq U^n$ are sufficiently independent 
from (i.e. not tightly coupled, or interdependent, with) each other. Some 
examples are in order.
Given two subsets $X, Y \subseteq U$, then $X \times Y$ is an instance
of relations whose two components are completely independent from each other.
On the other hand, the equality relation $\{ (x,x) : x \in U \}$ is an example 
of relations whose two components are tightly coupled. In this paper, we will
adopt the commonly studied notion of component-independence\footnote{Also called
variable-independence.} (e.g.~\cite{Lib00,VBN+17,Berstel,Valiant75}) in a 
relation $R \subseteq U^n$ that lies between the extremes as exemplified in the
above examples, i.e., that $R$ is expressible as a \emph{finite union} 
$\bigcup_{i=1}^r X_{i,1} \times \cdots \times X_{i,n}$ of products, where
each $X_{i,j}$ is expressible in the same language $\Lang$ (e.g. a logic or a 
machine model) wherein $R$ is expressed. 

Why should one care about
monadic decomposable relations? The main reason is that applying appropriate 
monadic restrictions could make an undecidable problem decidable, 
and in general turn a difficult problem into one more amenable to analysis. 
Several examples are in order. Firstly,
the well-known cartesian abstractions in abstract interpretation \cite{CC79} 
overapproximate 
the set  $R \subseteq U^n$ of reachable states at a certain program point
by a relation $R' \subseteq X_1 \times \cdots \times X_m$
such that $R \subseteq R'$.
%, which is essentially done by forgetting the
%relationships between elements in the tuples in $R$. 
Having $R'$ instead of $R$ 
sometimes allows a static analysis tool to prove correctness properties about a 
program that is otherwise difficult to do with only $R$. Another example
includes restrictions to monadic predicates in undecidable logics that result
in decidability, e.g., monadic
first-order logic and extensions (\cite{Boolos-book,BGG97,BDT10}), as well
as monadic second-order theory of successors \cite{BGG97}. Monadic 
decomposability also found applications in more efficient variable elimination
in constraint logic programming (e.g.~\cite{Imbert94}), as well as constraint
processing algorithms for constraint database queries
(e.g.~\cite{Lib00,Libkin-CD-book}).
Finally, monadic decompositions in the context of SMT (Satisfiability Modulo 
Theories), whose study was recently initiated in \cite{VBN+17}, have 
numerous applications, including constraint solving over strings 
\cite{VBN+17,popl19}.

The focus of this paper is to revisit a classical problem of determining monadic
decomposability of \defn{regular relations}, which are also known as 
\defn{synchronized rational relations} \cite{FS93,BLSS03,Bl99}. The study of
classes of relations over words definable by different classes of multi-tape 
(finite) automata is by now a well-established subfield of formal language 
theory.  This study was initiated by Elgot, Mezei, and Nivat in the 1960s 
\cite{EM,Nivat}; also see the surveys \cite{Berstel,Choffrut-survey}.
In particular, we have a strict hierarchy of classes of relations as follows:
recognizable relations, synchronized rational relations, deterministic rational 
relations, and rational relations. All these classes over unary relations (i.e.
languages) coincide with the class of regular languages. \defn{Rational 
relations} are
relations $R \subseteq (\Sigma^*)^n$ definable by multi-tape 
automata, where the tape heads move from left to right (in the usual way for
finite automata) but possibly at different speeds (e.g. in a transition, the 
first head could stay at the same position, whereas the second head moves to
the right by one position). \defn{Deterministic rational relations} are simply
those rational relations that can be described by deterministic multi-tape
automata. So far, the heads of the tapes can move at different speeds.
\defn{Regular relations} (a.k.a. \defn{synchronized rational relations}) are
those relations that are definable by multi-tape automata, all of whose heads 
move to the right in each transition. Unlike (non)deterministic rational 
relations, regular relations are extremely well-behaved, e.g., they are closed 
under first-order operations and, therefore, have decidable
first-order theories \cite{hodgson83}. Regular relations are also known to
coincide with those relations that are first-order definable over a universal
automatic structure \cite{BLSS03,Bl99}; equivalently, those relations that are
definable by finite-set interpretations in the weak-monadic theory of one
successor (WS1S) \cite{CL07}.
Finally, the weakest class of relations in the hierarchy are 
\defn{recognizable relations}:
those relations that are definable as a finite union of products of
regular languages or, equivalently, relations that can be defined as a boolean
combination of regular constraints (i.e. atomic formulas of the form $x \in L$, 
where $L$ is a regular language, asserting that the word $x$ is in $L$).
Recognizable relations are, therefore, those relations definable by multi-tape
automata that exhibit monadic decomposability.

One of the earliest results on deciding whether a relation is monadic decomposable 
follows from Stearns in 1967 \cite{Stearns} and the characterization of a binary
relation $R\subseteq A^*\times B^*$ by
$L_R=\{\textsf{rev}(u)\#v \mid (u,v)\in R\}$, where $\textsf{rev}(u)$
is the mirror image of $u$. In \cite{CCG06} it was proven that
$L_R$ is a regular language if and only if $R$ has a monadic
decomposition and if $R$ is a deterministic rational relation, then $L_R$ is a
deterministic context-free language.
Due to this characterization, Stearns's result implies that
whether a deterministic $n$-ary rational relation is monadic decomposable 
(i.e. recognizable) is decidable in the case when $n=2$. Shortly thereafter, 
Fischer and Rosenberg \cite{FR68} showed that the same problem is unfortunately 
undecidable for the full class of binary rational relations. A few years later
Valiant \cite{Valiant75} improved the upper bound complexity for the case solved 
by Stearns to double exponential-time. This is still the best known upper bound
for the monadic decomposability problem for deterministic binary rational 
relations to date and, furthermore, no specific lower bounds are known. More
recently Carton \emph{et al.} \cite{CCG06} adapted the techniques from
\cite{Stearns,Valiant75} to show that this decidability extends
to general $n$-ary relations, though no complexity analysis was provided.
The problem of monadic decomposability for regular relations has also been
studied in the literature. Of course decidability with a double exponential-time
upper bound for the binary case follows from \cite{Valiant75}. In 2000 Libkin 
\cite{Lib00} gave general conditions for monadic decomposability for first-order
theories, which easily implies decidability for monadic decomposability for 
general 
$k$-ary regular relations. This is because regular relations are simply those
relations that are definable in a universal automatic structures
\cite{BLSS03,Bl99}.
The result of Libkin was not widely known in the automata theory community and 
in fact the problem was posed as an open  problem in French version of 
\cite{Sakarovitch09} in 2003 and later on, Carton \emph{et al.} \cite{CCG06} 
provided a double-exponential-time algorithm for 
deciding whether an $n$-ary regular relation is monadic decomposable.
More precisely, even though
it was claimed in the paper that the algorithm runs in single-exponential
time, it was noted in a recent paper by L\"{o}ding and Spinrath \cite{LS17,LS19}
(with which the authors of \cite{CCG06} also agreed, as claimed in \cite{LS19}) that
the algorithm actually runs in double-exponential time. L\"{o}ding and Spinrath
\cite{LS17,LS19} gave a single-exponential-time algorithm (inspired by techniques 
from \cite{Valiant75}) for monadic decomposability of \emph{binary} regular 
relations.

\paragraph*{Contributions}
In this paper we provide the precise complexity of monadic decomposability of
regular relations, 
%finally 
closing the open questions left by Carton \emph{et
al.} \cite{CCG06} and L\"{o}ding and Spinrath \cite{LS17,LS19}. In particular, we
show the following.

\begin{theorem}\label{thm:main}
%Let $R$ be an $n$-ary regular relation. 
    Deciding whether a given regular relation $R$ is monadic decomposable is 
    $\NLOGSPACE$-complete (resp.~$\PSPACE$-complete), if $R$ is given by a 
    DFA (resp.~an NFA). 
\end{theorem}
The lower bounds hold already for binary 
relations (Lemma~\ref{thm:lb} and Lemma~\ref{lem:lbdfa} in Section~\ref{sec:hardness}).
To prove the upper bounds, we first prove the upper bounds for binary relations
(Lemma~\ref{lem:part1} in Section~\ref{sec:binary})
and then extend them to $n$-ary relations for any 
%fixed
given 
$n > 2$
(Lemma~\ref{thm:nary} in Section~\ref{sec:general}).
%The upper bound holds for general $k$-ary relations, whereas the lower bound
%holds for binary relations.

The existing proof techniques (e.g. in \cite{CCG06,LS19,Lib00}) for deciding
monadic decomposability typically aim for finding proofs that the relations are
monadic decomposable. In contrast, our proof technique relies on finding a proof
that a relation is \emph{not} monadic decomposable. As a brief illustration,
suppose we want to show that the regular relation $R = \{ (v,v) : v \in 
\Sigma^* \}$
%, which is a regular relation, 
is not monadic decomposable. 
We define an equivalence relation $\sim\ \subseteq \Sigma^* \times
\Sigma^*$ as 
\begin{align*}
    x \sim y := \forall z( [R(x,z) \leftrightarrow R(y,z)] \wedge 
                            [R(z,x) \leftrightarrow R(z,y)]).
\end{align*}
This relation is regular since regular relations are closed under first-order
operations \cite{Sakarovitch09} (a fact that was also used in \cite{CCG06}), but the size of the 
automaton
for this relation is unfortunately quite large; see \cite{LS17} for detailed discussion.
Therefore, we will only use the
complement $\not\sim$, which has a substantially smaller
representation: polynomial (resp.~exponential) size if $R$ is given as a DFA
(resp.~an NFA).
Now, that $R$ is not monadic decomposable amounts to the existence of an 
$\omega$-sequence $\sigma = \{v_i\}_{i \in \N}$ of words such that 
$v_i \not\sim v_j$ for each pair $i,j \in \N$.
By applying the pigeonhole principle and K\"onig's lemma, we will first construct a nicer
sequence $\alpha$ (see the top half of Figure~\ref{fig:newstairs}) and then 
by exploiting Ramsey Theorem over
infinite graphs, we will show that there is an even nicer sequence $\alpha'$ (see 
the bottom half of Figure \ref{fig:newstairs}), where the automaton for $\not\sim$ synchronizes its states
in particular points of the computation, no matter which pair of words from
the sequence is being read. Moreover, we prove that one of the synchronizing 
states has a pumping property. This leads to our $\NLOGSPACE$ algorithm as
we can guess
the synchronizing states and verify that there is an accepting run that can be pumped.
This technique was inspired by a technique for proving recurrent 
reachability in regular model checking \cite{anthony-thesis,TL08}.

\OMIT{
\anthony{The following paragraph seems out of place.}
%Another difference from previous results has been alluded before. 
Our techniques are in contrast to the techniques employed in 
\cite{Stearns, Valiant75, LS19}, which relied on characterization of a relation 
$R$ using the language $L_R$. Essentially,
their main technical results have been finding suitable machinery that is able to
decide whether $L_R$ is regular or not. Our technique is different as we are looking
at the relation directly and, in the end, use standard techniques to check
reachability in finite automata.
}

The exponential-time upper bound for the binary case from L\"{o}ding and 
Spinrath \cite{LS19} (which is inspired by the techniques used by Stearns 
\cite{Stearns} and Valiant \cite{Valiant75}) relied on characterization of a 
relation $R$ using the language 
$L_R=\{\textsf{rev}(u)\#v \mid (u,v)\in R\}$ and used a suitable machinery that is 
able to decide whether $L_R$ is regular or not. 
%Our technique is different as we are looking
%at the relation directly and, in the end, use standard techniques to check
Their result is not easily extensible to $n$-ary relations as 
the encoding of a binary rational relation as a context-free language $L_R$ 
does not generalize to $n$-ary relations.
%there is no known characterization of relations as a language.
In Section \ref{sec:general}, we show that proving monadic decomposability for an $n$-ary 
regular relation is $\LOGSPACE$-reducible to testing whether linearly many induced 
binary relations are monadic decomposable. 

We conclude in Section~\ref{sec:conc} with some perspectives from formal
verification and a future research direction.

%Note that we can think of
%$\sigma$ as a complete graph (i.e. clique) with infinitely many nodes. 

\section{Preliminaries}
\label{sec:prelim}
A finite alphabet is denoted by $\Sigma$ and the free monoid it generates 
by $\Sigma^*$. That is, $\Sigma^*$ consists of all finite words over $
\Sigma$. The empty word is $\varepsilon$. We denote by $|w|$ the length of word
$w\in\Sigma^*$. We have that $|\varepsilon|=0$.  The word $u\in\Sigma^*$ is a
\emph{prefix} of $w\in\Sigma^*$ if $w=uv$ for some $v\in\Sigma^*$. We denote this
by $u\leq w$. We also write $v=u^{-1}w$, when $u$ is a prefix of $w$, to state that $v$ is the suffix of $w$ 
that is obtained after prefix $u$ is removed. Sometimes we
want to consider a suffix of $w$ after a prefix of particular length is removed
without specifying the actual prefix as defined above. To this end, we define partial function
$\sigma: \Sigma^*\times \mathbb{N}\rightarrow \Sigma^*$ such that $\sigma(w,i)=v$,
where $w=uv$ for some $u\in\Sigma^*$ such that $|u|=i$. In particular, for $u\leq w$,
$\sigma(w,|u|)=u^{-1}w$. Similarly, we define partial function $\tau: \Sigma^*\times\mathbb{N}\rightarrow\Sigma^*$
such that $\tau(w,i)=u$, where $|u|=i$ and $u\leq w$.

In this paper we study relations $R\subseteq \Sigma^*\times\cdots\times\Sigma^*$
with particular structural properties.
Namely, \emph{monadic decomposable} relations that are a finite union of
direct products of regular languages, and \emph{regular} relations
defined by $n$-tape finite automata, where the heads move in synchronized
manner. See, for example, \cite{Sakarovitch09} for more details on such relations.

\begin{definition}
An $n$-ary relation $R\subseteq \Sigma^*\times\cdots\times\Sigma^*$ is a \emph{monadic decomposable} 
relation iff it is of the form 
$\bigcup_{i=1}^{m}(X_{1,i}\times\cdots\times X_{n,i})$, where $m$ is finite and
each $X_{j,i}\subseteq \Sigma^*$ is a regular language.
\end{definition}
As mentioned earlier, this can be intuitively seen as the components of $R$ being independent in some sense.  
Note that in the literature, monadic decomposable relations are sometimes
called \emph{recognizable}. 
The monadic decomposable relations can be defined using multi-tape automata
as is done, e.g., in \cite{CCG06}. The above definition is more suitable
for our considerations.

Let $\pad$ be a fresh symbol not found in $\Sigma$. We use it to pad
words in a relation $R\subseteq\Sigma^* \times\cdots\times\Sigma^*$ in order
for each component to be of the same length. Formally, a tuple $(w_1,\ldots,w_n)$ 
is transformed into $(w_1\pad^{\ell_1},\ldots,w_n\pad^{\ell_n})$, where 
$\ell_i = -|w_i|+\max_{1\leq j\leq n}|w_j|$ for each $i=1,\ldots,n$. We extend this
to the relation $R_\pad$ in the expected way. We also denote $\Sigma\cup\{\pad\}$ by $\Sigma_\pad$.
An {\em $n$-tape automaton over alphabet $\Sigma_\pad$} is a tuple 
$(Q,\rightarrow_{\A},q_0,F)$,
where 
%$\Sigma_\pad$ is a finite non-empty alphabet, 
$Q$ is the
finite set of states, $q_0$ is the initial state, $F$ is the set of final
states, and 
$\rightarrow_\A~\subseteq Q\times (\Sigma_\pad)^n \times\,\mathcal{P}(Q)$. 

\begin{definition}
An $n$-ary relation $R \subseteq \Sigma^* \times \cdots \times \Sigma^*$ is \emph{regular} iff 
$R_\pad$ is recognized by some $n$-tape automaton $\A_\pad$ 
over alphabet $\Sigma_\pad$.
\end{definition}

That is, in a regular relation the $n$ heads of the automaton are moving in
synchronized manner and the $n$-tuple of symbols seen 
determines the state transition. Naturally, the
state transition can be deterministic or non-deterministic. 
We say that a regular relation is defined by an NFA if the underlying $n$-tape 
automaton is non-deterministic, otherwise we say that the relation is
defined by a DFA.
Note that in the literature, regular relations are sometimes called 
\emph{synchronous rational} or \emph{automatic} relations.

We recall a useful characterization from \cite{CCG06}.
Consider an $n$-ary regular relation $R \subseteq \Sigma^*\times\cdots \times \Sigma^*$. 
For each $j=1,\ldots,n-1$, let $\sim_j$ be the following induced
equivalence relation:
\begin{align*}
\begin{multlined}
(u_1,\ldots,u_j) \sim_j (v_1,\ldots,v_j) \, := \, \forall (w_{j+1},\ldots,w_n)\in \Sigma^* \times \cdots \times \Sigma^* \text{ we have that} \\
\qquad \qquad (u_1,\ldots,u_j,w_{j+1},\ldots,w_n)\in R \, \Longleftrightarrow \, (v_1,\ldots,v_j,w_{j+1},\ldots,w_n)\in R \text{ and} \\
 \qquad \qquad (w_{j+1},\ldots,w_n,u_1,\ldots,u_j)\in R \, \Longleftrightarrow \, (w_{j+1},\ldots,w_n,v_1,\ldots,v_j)\in R. 
\end{multlined}
\end{align*}

\begin{lemma}[\cite{CCG06}]\label{lm:equivrel}
The $n$-ary regular relation $R$ is monadic decomposable iff  $\sim_j$ has finite index for
each $j = 1,\dots,n-1$. That is, there are finitely many equivalence classes over $\sim_j$.

In other words, $R$ is not monadic decomposable iff for some $j=1,\ldots,n-1$, there is an infinite
sequence $\{u_i\}_{i\geq0}$, where each $u_i$ is a $j$-tuple of words, such that for each $0 \leq i < \ell$ it is the case that 
$u_i \neq u_\ell$ and 
$u_i \not\sim_j u_\ell$.
\end{lemma}

In Section~\ref{sec:binary}, we focus on binary relations for which we simplify the
notation as there is only one possible value of $j$. We write $\sim$ instead
of $\sim_j$ and $R^{\not\sim}$ for the binary regular relation
\begin{multline*}%\label{eq:sim}
R^{\not\sim}(w,w') \, := \, \exists u \,\big((R(w,u) \wedge \neg R(w',u)) \, \vee \, (\neg R(w,u) \wedge R(w',u)) \, \vee \, \\ 
(R(u,w) \wedge \neg R(u,w')) \, \vee \, (\neg R(u,w) \wedge R(u,w'))\big).
\end{multline*}
That is, $R^{\not\sim}$ consists of all words $w,w'\in\Sigma^*$ for which there exists a word $u\in\Sigma^*$ such that one of $R(w,u)$ and $R(w',u)$ is accepted while the other is not, or one of $R(u,w)$ and $R(u,w')$ is accepted while the other is not.

We assume that the reader is familiar with complexity classes and logarithmic space
reductions via logarithmic space transducers; see for example \cite{Sipser2006}.

%In the proof of the upper bound, we construct a B\"uchi automaton from a regular
%relation $R$ such that the language of the automaton is empty if and only if
%$R$ is monadic-decomposable. Let us define B\"uchi automata next.

%\begin{definition}
%A B\"uchi automaton $\B$ is a tuple $(\Sigma,Q,\rightarrow_\B,q_0,F)$, where
%$\Sigma$ is a finite non-empty alphabet, $Q$ is the finite set of states, $q_0$
%is the initial state, $F$ is the set of final states, and 
%$\rightarrow_\B\subseteq Q\times \Sigma\times \mathcal{P}(Q)$ is the transition
%function. An infinite word $w\in\Sigma^\omega$ is accepted by $\B$ if there is an
%infinite sequence of transitions in $\B$ reading $w$ such that some state $f\in F$
%is visited infinitely often. The language of infinite words accepted by $\B$ is
%denoted by $L(\B)$. 	
%\end{definition}

\section{Hardness of deciding monadic decomposability of regular relations}
\label{sec:hardness}
In this section, we consider binary regular relations given by NFA and
provide a $\PSPACE$ lower bound for deciding if
such a relation is monadic decomposable. Then, we prove that the same problem for
DFA is $\NLOGSPACE$-hard.

\begin{lemma}\label{thm:lb}
The problem of deciding whether a binary regular relation given by an NFA is monadic decomposable is $\PSPACE$-hard.
\end{lemma}
\begin{proof}
We give a logarithmic space reduction from the universality problem for NFA, which is 
%well-known to be 
$\PSPACE$-hard \cite{MS72}.
Recall that in this problem, we are asked 
to decide 
whether $L(\A) = \Sigma^*$
given an NFA $\A$ over 
%alphabet 
$\Sigma$.

Let $\mathcal{A}$ be an NFA over alphabet $\Sigma$, and let $\{\sep\}$ be a fresh symbol that we will use as a separator symbol. We assume that $\sep\neq\pad$.
We construct relation $R=R_1\cup R_2$ using the language $L$ of $\A$, where
\begin{align*}
R_1 = \{(u,u)\mid u \in (\Sigma \cup \{\sep\})^*\} \qquad \text{and} \qquad
R_2 =  (L \cdot \{\sep\} )^* \times (\Sigma^* \cdot \{\sep\} )^*.
\end{align*}
Intuitively, $R_1$ contains all pairs $(w_1,w_2)$ such that
$w_1=w_2=u_0\sep u_1\sep \cdots \sep u_n \sep$, where $u_i\in\Sigma^*$, and $R_2$
contains all pairs $(w_1,w_2)$ such that
$w_1=v_0\sep v_1\sep \cdots\sep v_m \sep$, where $v_i\in L$, and 
$w_2=u'_0\sep u'_1\sep \cdots \sep u'_n \sep$, where $u'_i\in\Sigma^*$.
It is easy to construct an NFA that recognizes $R$ in $\LOGSPACE$. 
Next we show that $L = \Sigma^*$ iff $R$ is monadic decomposable. 

Assume first that $L = \Sigma^*$. Then 
$R_1 \subseteq R_2$, and thus
$R = (\Sigma^* \cdot \{\sep\} )^* \times (\Sigma^* \cdot \{\sep\} )^*$
which has a trivial monadic decomposition. 

For the other direction, assume that $R$ is monadic decomposable, i.e., $R = \bigcup_{i=1}^n (A_i \times B_i)$ 
for some regular languages $A_i$, $B_i$. Let $w \in \Sigma^*$. We show that
$w \in L$ as well. Consider a
set $\{((w\sep)^i,(w\sep)^i)\mid i=1,\ldots, n+1\}\subseteq R_1\subseteq R$.
By the pigeonhole principle, there are two elements $((w\sep)^j,(w\sep)^j)$ and
$((w\sep)^k,(w\sep)^k)$
that belong to the same component of $\bigcup_{i=1}^n (A_i \times B_i)$, say to $A_1\times B_1$. Therefore,
$(w\sep)^{j} \in A_1$ and $(w\sep)^{k} \in B_1$, and hence
their direct product, $((w\sep)^{j},(w\sep)^k)$, is in $A_1\times B_1\subseteq R.$
Recall that $R=R_1\cup R_2$. 
Clearly, $((w\sep)^{j},(w\sep)^k)\notin R_1$ as the lengths of the two words are
different. It follows that $((w\sep)^{j},(w\sep)^k) \in R_2$ and hence 
$(w\sep)^{j} \in (L \cdot \{\sep\})^*$. This implies that $w \in L$. 
%, which 
%completes the proof.
\end{proof}

\begin{lemma}\label{lem:lbdfa}
The problem of deciding whether a binary regular relation given by a DFA is
monadic decomposable is $\NLOGSPACE$-hard.
\end{lemma}
\begin{proof}
We prove the hardness by a logarithmic space reduction from the reachability problem for directed
acyclic graphs, which is an $\NLOGSPACE$-hard problem \cite{Jones75}. Let $G$ be a
directed acyclic graph, $s,t$ two vertices of $G$, and we are asked whether $t$ is 
reachable from $s$. Let $d$ be the degree of the graph.

We construct DFA $\A$ out of $G$. The states are the vertices of $G$ together with
new sink state $\bot$. The initial state is $s$ and the final state is $t$.
The alphabet is $\Sigma=\{a_1,\ldots,a_d,a_{d+1}\}$.
Let $s_1$ be a vertex of outdegree $d'\leq d$ that has edges to $t_1,\ldots t_{d'}$.
For each edge from $s_1$ to $t_i$, we add a transition $(s_1,(a_i,a_i),s_i)$.
Finally, we add a self-loop $(t,(a_{d+1},a_{d+1}),t)$ and transitions
$(s',(a_i,a_j),\bot)$ for every state $s'$ and all $i\neq j$.

Observe that $\A$ is a DFA by our choice of labels on transitions. Moreover, the
relation defined by $\A$ consists of words $(u,u)$ for some $u\in\Sigma^*$.
The relation has finitely many different elements if and only if $t$ is not
reachable from $s$. Recall that all relations with finitely many elements are
monadic decomposable. On the other hand, if $t$ is reachable, then the relation
is not monadic decomposable, which completes the proof.
\end{proof}
\section{Deciding monadic decomposability of binary regular relations}
\label{sec:binary}
In this section we prove our main technical result.

\begin{lemma} \label{theo:ptime} 
There is an $\NLOGSPACE$ algorithm that takes as input an NFA for $R^{\not\sim}$, where $R$ is a binary regular relation, 
and decides whether $R$ is monadic decomposable. 
\end{lemma}  

We start by defining some notation. 
We assume any binary regular relation $R^{\not\sim}$ to be given as an NFA with 
set of states $Q$. The {\em $R^{\not\sim}$-type} of a pair $(w_1,w_2)$ of words
over $\Sigma$ is an element of the {\em transition monoid}. Recall that the transition
monoid transforms any given state $q \in Q$ to a set $Q' \subseteq Q$ of states 
when reading $(w_1,w_2)$. We denote this by $R^{\not\sim}_{w_1,w_2}(q)$ for each
$q\in Q$. We write ${\sf types}({R^{\not\sim}})$ for the set of all $R^{\not\sim}$-types. 

Consider an infinite sequence $\{w_i\}_{i \geq 0}$ of words 
over $\Sigma$ as defined in Lemma~\ref{lm:equivrel}.
Additionally, we assume that the words in the sequence are of strictly increasing
length and that for each $i>0$ the words $w_i$ and $w_{i+1}$ have a common prefix of length $|w_{i-1}|$.
That is, $w_i$ can be written as $\beta_0\cdots\beta_{i-1}\alpha_i$,
where each $\beta_j$ and $\alpha_i$ is a non-empty word.
To simplify notation, we denote $\rho(w_i)=\beta_0\cdots\beta_i$. 
That is, $\rho(w_i)$ is of length $|w_i|$ and is a prefix of $w_j$, for each $0 \leq i < j$. 
%(In fact, it is the prefix of $w_j$ of the same length as $w_i$). 
We will show how to construct such sequence in Proposition~\ref{prop:char}.
The words $w_i$, $w_j$ and $w_k$ are illustrated in the top of Figure~\ref{fig:Cij}.

%We write $w_j^i$ for the word that satisfies $w_j = w'_{i} w^i_j$. 
%Let us consider the following coloring of the infinite clique over $\mathbb{N}$. 
%The set of colors is $\Gamma_R \times \Gamma_R \times \Gamma_R \times \Gamma_R$.  
%Let us assume that $0 \leq i < j$. 
With each pair $(i,j)$, where $i < j$, we associate the following quinary tuple over ${\sf types}(R^{\not\sim})$: 
%is colored with the tuple 
\begin{align*}
\mathfrak{C}_{i,j} \ = \ \big(\, R^{\not\sim}_{w_i,\rho(w_i)}, \, R^{\not\sim}_{\rho(w_i),\rho(w_i)},\,R^{\not\sim}_{\sigma(w_j,|w_i|),\sigma(\rho(w_j),|w_i|)},\, R^{\not\sim}_{\varepsilon,\sigma(w_j,|w_i|)},\, R^{\not\sim}_{\varepsilon,\sigma(\rho(w_j),|w_i|)} \,\big).
\end{align*}
Intuitively, the first component corresponds to the computation of
$(\beta_0\cdots\beta_{i-1}\alpha_i,\beta_0\cdots\beta_{i-1}\beta_i)$, the second
to $(\beta_0\cdots\beta_{i-1}\beta_i,\beta_0\cdots\beta_{i-1}\beta_i)$ needed in
order to compute the third component, $(\beta_{i+1}\cdots\beta_{j-1}\alpha_j,\beta_{i+1}\cdots\beta_{j-1}\beta_j)$.
The final two components are used to compute the set of states reachable after the
whole word in the first component is read. That is
$(\pad^{|\beta_{i+1}\cdots\beta_{j-1}\alpha_j|},\beta_{i+1}\cdots\beta_{j-1}\alpha_j)$ and 
$(\pad^{|\beta_{i+1}\cdots\beta_{j-1}\beta_j|},\beta_{i+1}\cdots\beta_{j-1}\beta_j)$. 
See Figure~\ref{fig:Cij} for a pictorial depiction.

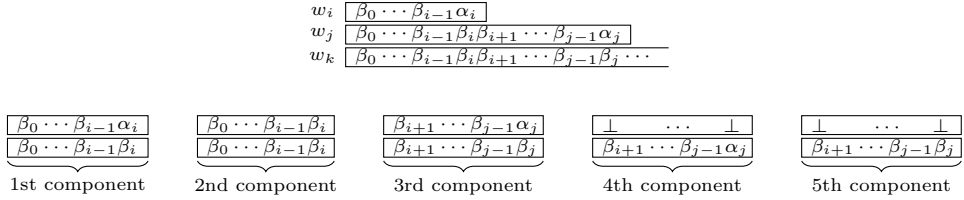
\begin{figure}[htb]
\centering
\begin{tikzpicture}%[scale=1,every node/.style={scale=0.97}]

\draw (-0.05,0) rectangle (1.8,-0.25);
\draw (-0.05,-0.30) rectangle (3.7,-0.55); 
\draw (4.2,-0.6) -- (-0.05,-0.6) -- (-0.05,-0.85) -- (4.2,-0.85);
\node at (-0.35,-0.125) {\scriptsize $w_i$}; 
\node at (-0.35,-0.425) {\scriptsize $w_j$}; 
\node at (-0.35,-0.725) {\scriptsize $w_k$};

\node[right] at (-0.05,-0.125) {\scriptsize $\beta_0\cdots\beta_{i-1}\alpha_i$};
\node[right] at (-0.05,-0.425) {\scriptsize $\beta_0\cdots\beta_{i-1}\beta_i\beta_{i+1}\cdots\beta_{j-1}\alpha_j$};
\node[right] at (-0.05,-0.725) {\scriptsize $\beta_0\cdots\beta_{i-1}\beta_i\beta_{i+1}\cdots\beta_{j-1}\beta_j\cdots$};

\draw (-4.5,-1.5) rectangle (-2.65,-1.75);
\draw (-4.5,-1.8) rectangle (-2.65,-2.05);
\node[right] at (-4.5,-1.625) {\scriptsize $\beta_0\cdots\beta_{i-1}\alpha_i$};
\node[right] at (-4.5,-1.925) {\scriptsize $\beta_0\cdots\beta_{i-1}\beta_i$};
\draw[decorate,decoration={brace,amplitude=4pt}] (-2.65,-2.1) -- (-4.5,-2.1) node[below, pos=0.5,yshift=-0.1cm,align=left]{\scriptsize 1st component};

\draw (-2,-1.5) rectangle (-0.2,-1.75);
\draw (-2,-1.8) rectangle (-0.2,-2.05);
\node[right] at (-2,-1.625) {\scriptsize $\beta_0\cdots\beta_{i-1}\beta_i$};
\node[right] at (-2,-1.925) {\scriptsize $\beta_0\cdots\beta_{i-1}\beta_i$};
\draw[decorate,decoration={brace,amplitude=4pt}] (-0.2,-2.1) -- (-2,-2.1) node[below, pos=0.5,yshift=-0.1cm,align=left]{\scriptsize 2nd component};

\draw (0.45,-1.5) rectangle (2.55,-1.75);
\draw (0.45,-1.8) rectangle (2.55,-2.05);
\node[right] at (0.45,-1.625) {\scriptsize $\beta_{i+1}\cdots\beta_{j-1}\alpha_j$};
\node[right] at (0.45,-1.925) {\scriptsize $\beta_{i+1}\cdots\beta_{j-1}\beta_j$};
\draw[decorate,decoration={brace,amplitude=4pt}] (2.55,-2.1) -- (0.45,-2.1) node[below, pos=0.5,yshift=-0.1cm,align=left]{\scriptsize 3rd component};

\draw (3.2,-1.5) rectangle (5.3,-1.75);
\draw (3.2,-1.8) rectangle (5.3,-2.05);
\node[right] at (3.2,-1.625) {\scriptsize $\pad\qquad\cdots\ \quad \pad$};
\node[right] at (3.2,-1.925) {\scriptsize $\beta_{i+1}\cdots\beta_{j-1}\alpha_j$};
\draw[decorate,decoration={brace,amplitude=4pt}] (5.3,-2.1) -- (3.2,-2.1) node[below, pos=0.5,yshift=-0.1cm,align=left]{\scriptsize 4th component};

\draw (5.95,-1.5) rectangle (8.05,-1.75);
\draw (5.95,-1.8) rectangle (8.05,-2.05);
\node[right] at (5.95,-1.625) {\scriptsize $\pad\qquad\cdots\ \quad \pad$};
\node[right] at (5.95,-1.925) {\scriptsize $\beta_{i+1}\cdots\beta_{j-1}\beta_j$};
\draw[decorate,decoration={brace,amplitude=4pt}] (8.05,-2.1) -- (5.95,-2.1) node[below, pos=0.5,yshift=-0.1cm,align=left]{\scriptsize 5th component};
\end{tikzpicture}
\caption{\label{fig:Cij} Correspondence between components of $\mathfrak{C}_{i,j}$ and parts of computation on $w_i$, $w_j$ and $w_k$, where $i<j<k$.}
\end{figure}

We can then establish the following important proposition. 
Consider an infinite sequence of words that are pairwise from different equivalence
classes as in Lemma~\ref{lm:equivrel}. We show next that we can
extract an infinite subsequence with additional structural properties. Perhaps
the most important property is that $\mathfrak{C}_{i,j}$ is the same for all $i,j$.
This subsequence will allow us to prove the main lemma.

\begin{proposition} \label{prop:char} 
A binary regular relation $R$ over $\Sigma^* \times \Sigma^*$ is not monadic
decomposable iff there are infinite sequences $\{u_i\}_{i \geq 0}$,
$\{\gamma_i\}_{i \geq 0}$, and $\{\delta_i\}_{i \geq 0}$ of words over $\Sigma$ and a quinary tuple
 $\mathfrak{C}$ over ${\sf types}(R^{\not\sim})$ such that for each $i \geq 0$
it is the case that 
\begin{enumerate} 
\item $|\gamma_i| = |\delta_i| > 0$, 
\item $u_i = \delta_0 \cdots \delta_{i-1} \gamma_i$,  
\item $(u_i,u_j) \in R^{\not\sim}$, for each $j > i$, and 
\item $\mathfrak{C}_{i,j} = \mathfrak{C}$, for each $j > i$. 
\end{enumerate} 
\end{proposition} 

%The proposition is proven in two steps. First, we use the pigeonhole principle and
%K\"onig's Lemma to construct a sequence, where the words have particular structure; namely, for each $i$ it is the case that $w_i=\beta_0\cdots\beta_{i-1}\alpha_i$ for some $\beta_j,\alpha_j\in \Sigma^*$. Then by applying Infinite Ramsey's Theorem, we extract a subsequence where $\mathfrak{C}_{i,j}=\mathfrak{C}$ for each $i,j$. See Figure~\ref{fig:newstairs} for pictorial depiction and Appendix for the formal proof.

\begin{proof}
By Lemma~\ref{lm:equivrel}, the existence of such sequences directly 
implies that the relation is not monadic decomposable. Assume then that $R$ is not monadic decomposable. By 
Lemma~\ref{lm:equivrel}, there exists a sequence $\{v_i\}_{i\geq0}$ such that 
$R^{\not\sim}(v_j,v_\ell)$ for all $j\neq \ell$.
It remains to show how to construct the three sequences satisfying the additional
properties from $\{v_i\}_{i\geq0}$.
First, we construct an auxiliary sequence $\{w_i\}_{i\geq0}$ in the following way. 
Let $v_j$ be the first non-empty word of $\{v_i\}_{i\geq0}$.
Denote $v_j=w'_0=\alpha_0$. Consider prefixes of $v_i$ of
length $|\alpha_0|$. Since $|\alpha_0|$ is finite and the sequence is
infinite, there exists a prefix that appears infinitely often by the pigeonhole
principle. Denote this prefix by $\beta_0$. Now we consider an infinite
subsequence $\{w'_i\}_{i\geq0}$ of $\{v_i\}_{i\geq0}$ where $w'_0=v_j$ 
and $w'_i$, where $i>0$, has $\beta_0$ as the proper prefix. We can write 
$w'_1=\beta_0\alpha_1$ and repeat the procedure. By K\"onig's Lemma, we can
always repeat the procedure and obtain the desired auxiliary sequence $\{w_i\}_{i\geq0}$
in the limit.

From Infinite Ramsey's Theorem, there is an infinite sequence $0 \leq \ell_0 < \ell_1 < \cdots$
and a tuple $\mathfrak{C} \in {\sf types}(R^{\not\sim})^5$ such that for each
$0 \leq i < j$ we have $\mathfrak{C}_{\ell_i,\ell_{j}} = \mathfrak{C}$. Namely,
we consider a complete infinite graph with natural numbers as vertices. An edge
between vertices $i$ and $j$ is coloured with
$\mathfrak{C}_{i,j}\in{\sf types}(R^{\not\sim})^5$. Now there is an infinite
clique coloured with $\mathfrak{C}$ which gives us our infinite sequence
$0 \leq \ell_0 < \ell_1 < \cdots$.

We then define the $u_i$s, $\gamma_i$s, and $\delta_i$s, for $i \geq 0$, as follows. 
\begin{itemize} 
\item $\gamma_0  = w_{\ell_0}$ and 
$\gamma_{i+1}$, for $i > 0$, is the word $\sigma(w_{\ell_{i+1}},|w_{\ell_i}|)$.  
%where $u$ is the prefix of $w_{\ell_{i+1}}$ that has the same length than $\delta_0 \dots \delta_{i-1} \gamma_i$.   
\item $\delta_i$ is defined as $\rho(\gamma_i)$. %, the word that satisfies $|\delta_i| = |\gamma_i|$ and 
%$w_{\ell_{i+1}} = v \delta_i v'$, where 
%$v = \delta_0 \dots \delta_{i-1}$. 
\item $u_i = \delta_0 \cdots \delta_{i-1} \gamma_i$, for each $i \geq 0$. 
\end{itemize} 
It is easy to see then that  $u_i = w_{\ell_i}$ and 
$\rho(u_{i}) = \delta_0 \cdots \delta_{i-1} \delta_i = \rho(w_{\ell_i})$, for each $i \geq 0$.  
Therefore, $\{u_i\}_{i \geq 0}$, $\{\gamma_i\}_{i \geq 0}$, $\{\delta_i\}_{i \geq 0}$,
and $\mathfrak{C}$ satisfy the conditions in the statement of the proposition. 
See Figure~\ref{fig:newstairs} for a pictorial depiction of the construction.
\end{proof}

In other words, by Proposition \ref{prop:char}, there is a sequence $\{u_i\}_{i\geq0}$
and a $\mathfrak{C}$ such that for each $i,j$, the runs on $R^{\not\sim}$ are
synchronized after $(\gamma_i,\delta_i)$, $(\delta_i,\delta_i)$,
$(\delta_i^{-1}\gamma_j,\delta_i^{-1}\delta_j)$,
$(\varepsilon,\delta_i^{-1}\gamma_j)$ and
$(\varepsilon,\delta_i^{-1}\delta_j)$ have been read.
In particular, the runs are synchronized in states of 
$R_{\gamma_i,\delta_i}^{\not\sim}$, $R_{\delta_i,\delta_i}^{\not\sim}$,
$R_{\delta_i^{-1}\gamma_j,\delta_i^{-1}\delta_j}^{\not\sim}$,
$R_{\varepsilon,\delta_i^{-1}\gamma_j}^{\not\sim}$ and
$R_{\varepsilon,\delta_i^{-1}\delta_j}^{\not\sim}$, respectively. %This is illustrated in Figure~\ref{fig:sameC}.
%See also Example~\ref{ex:stairs} in Appendix that illustrates the construction
%of Proposition~\ref{prop:char}.

\begin{figure}[htb]
\centering
\input{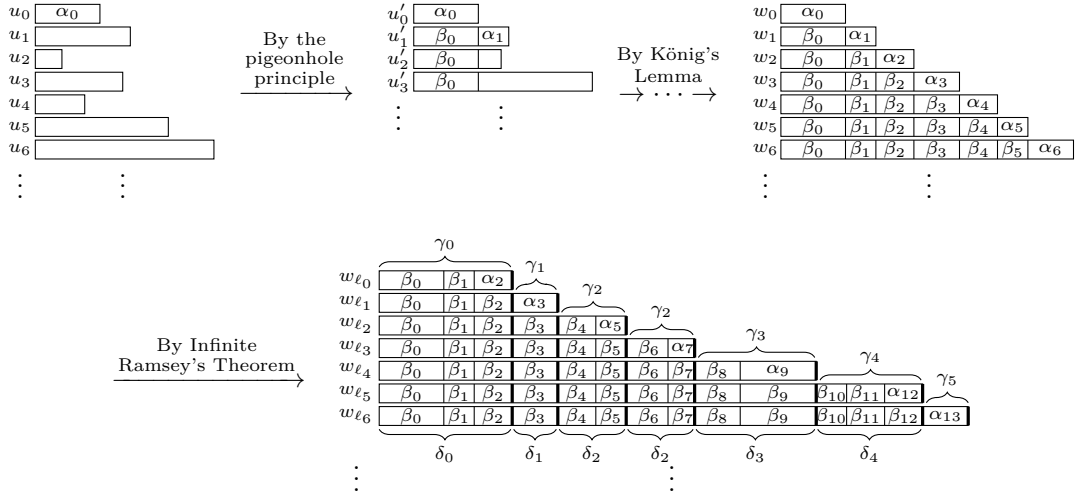}
\caption{\label{fig:newstairs} An illustration of construction of sequence $\{u_i\}_{i\geq0}$ of Proposition~\ref{prop:char} in two steps. Here 
$R^{\not\sim}(u_i,u_j)$, $R^{\not\sim}(u'_i,u'_j)$ and $R^{\not\sim}(w_i,w_j)$ for
every $i\neq j$.
Moreover as $\mathfrak{C}=\mathfrak{C}_{i,j}$,
the sets of states reachable after each $\delta_i$ and $\gamma_i$ are the same
(indicated by thick lines).
}
\end{figure}

\begin{example}\label{ex:stairs}
Let us illustrate the importance of Proposition~\ref{prop:char}. Consider relation
$R=\{(u,v)\in \Sigma^*\times\Sigma^*\mid u<v\}$, where $\Sigma=\{a,b\}$. This is a regular relation
as we can construct a DFA recognizing this relation; see the left automaton of
Figure~\ref{fig:prefix} where the sink state and the transitions to the sink state are omitted.
On the other hand, the relation does not have a monadic decomposition as
a rather convoluted sequence $\{u_i\}_{i\geq0}$ defined by 
\begin{align*}
u_i = \begin{cases} a^{4\lfloor\frac{i}{8}\rfloor}b & \text{if } i\equiv 0 \mod 8 \\
					a^{4\lfloor\frac{i}{8}\rfloor}a & \text{if } i\equiv 1 \mod 8 \\
					a^{4\lfloor\frac{i}{8}\rfloor}aa & \text{if } i\equiv 2 \mod 8 \\
					a^{4\lfloor\frac{i}{8}\rfloor}ab & \text{if } i\equiv 3 \mod 8 \\
					a^{4\lfloor\frac{i}{8}\rfloor}aab & \text{if } i\equiv 4 \mod 8 \\
					a^{4\lfloor\frac{i}{8}\rfloor}aaa & \text{if } i\equiv 5 \mod 8 \\
					a^{4\lfloor\frac{i}{8}\rfloor}aaaa & \text{if } i\equiv 6 \mod 8 \\
					a^{4\lfloor\frac{i}{8}\rfloor}aaab & \text{if } i\equiv 7 \mod 8
\end{cases}
\end{align*}
satisfies the properties of Lemma~\ref{lm:equivrel}.
Indeed, it is easy to see that for any $u_i$ and $u_j$, where $i\neq j$, $u_i\not\sim u_j$.
(In fact, for any $u,v\in\Sigma^*$, such that $u\neq v$, then $u\not\sim v$.)
An automaton for $R^{\not\sim}$ is presented on the right of Figure~\ref{fig:prefix}.

Next, we follow the steps of Proposition~\ref{prop:char}. First, we construct subsequence
$\{w_i\}_{i\geq0}$. In this sequence, $w_i=u_{2i}$.
Further, $\alpha_i=b$ if $i\equiv 0 \mod 2$ and $\alpha_i=a$ otherwise, and
$\beta_i=a$ for every $i$. Next, we highlight the importance of the second step
of the construction. Consider three words of the sequence, $v_1=0^n 1$, $v_2=0^{n'}$
and $v_3=0^{n''}1$, where $n<n'<n''$, and the runs on $(v_1,v_2)$ and $(v_2,v_3)$
in $R^{\not\sim}$.
The pair $(v_1,v_2)$ is accepted in state $q_2$, while $(v_2,v_3)$ is accepted
in $q_3$.
We can extract a subsequence where all accepting runs visit the same states,
e.g., $\{u_i\}_{i\geq0}$, where $u_0=00$ and $u_i=(00)^{i-1}$ for $i>0$. Now runs
on any pair of words will be accepted in $q_3$. That is, all pairs of words from
this sequence visit exactly the same states of $R^{\not\sim}$.
\end{example}

\begin{figure}
\centering
\begin{tikzpicture}[initial text=, initial distance=1.5ex,>=stealth, every state/.style={inner sep=2pt, minimum size=0pt}]
\node[state,initial] (q0) {$\phantom{q}$};
\node[state,accepting] (qf) [right = 1.5cm of q0] {$\phantom{q}$};

\path[->] (q0) edge node[above] {$(\bot,x)$} (qf)
		  (q0) edge[loop above] node[above] {$(x,x)$} (q0)
		  (qf) edge[loop above] node[above] {$(\bot,x)$} (qf);

\node[state,initial] (q0) [right = 2.5cm of qf] {$q_0$};
\node[state] (r1) [above right = 0.6cm and 1.5cm of q0] {$q_1$}; 
\node[state] (r2) [below right = 0.6cm and 1.5cm of q0] {$\phantom{q_0}$}; 
\node[state,accepting] (p1) [above right = 0.2cm and 2cm of r1] {$q_2$};
\node[state,accepting] (p2) [below right = 0.2cm and 2cm of r1] {$q_3$};
\node[state,accepting] (p3) [above right = 0.2cm and 2cm of r2] {$\phantom{q_0}$};
\node[state,accepting] (p4) [below right = 0.2cm and 2cm of r2] {$\phantom{q_0}$};

\path[->] (q0) edge node[above,sloped] {$(x,x)$} (r1)
		  (q0) edge node[below,sloped] {$(x,x)$} (r2)
		  (r1) edge[loop above] node[above] {$(x,x)$} (r1)
		  (r2) edge[loop below] node[below] {$(x,x)$} (r2)
		  
		  (r1) edge node[sloped,align=center,pos=0.6] {$(a,b)$ \\ $(b,a)$} (p1)
		  (r1) edge node[below,sloped] {$(\bot,x)$} (p2)
		  (r2) edge node[sloped,align=center] {$(a,b)$ \\ $(b,a)$} (p3)
		  (r2) edge node[below,sloped] {$(x,\bot)$} (p4)
		  
		  (p1) edge[loop right] node[right] {$(y,y')$} (p1)
		  (p2) edge[loop right] node[right] {$(\bot,x)$} (p2)
		  (p3) edge[loop right] node[right] {$(y,y')$} (p3)
		  (p4) edge[loop right] node[right] {$(x,\bot)$} (p4);

\end{tikzpicture}
\caption{\label{fig:prefix} Automata for binary regular relation $R$ (left) and
$R^{\not\sim}$ (right). Here $x\in\Sigma$ and $y,y'\in\Sigma_\bot$.}
\end{figure}

We can then prove the following crucial result. We assume here that  
$R$ is a binary regular relation over 
$\Sigma \times \Sigma$ such that $R^{\not\sim}$ is given as an NFA over 
$\Sigma \times \Sigma$ whose set of states is
$Q$. We further assume that $q_0$ is the initial state of $R^{\not\sim}$ and $F$ its set of 
final states. 

\begin{lemma} \label{theo:main} 
Relation $R$ is not monadic decomposable iff there 
are an infinite sequence $\{(x_i,y_i)\}_{i \geq 0}$ of pairs of words over $\Sigma$ 
and states $q,q',p,r \in Q$, such that $p \in F$, it is the case that 
$q \in R^{\not\sim}_{x_0,y_0}(q_0)$, and the following statements hold for each $i \geq 0$. 

\begin{enumerate}
\item It is the case  that $|x_{i}| = |y_i|$ and $y_i$ is a prefix of both 
$x_{i+1}$ and $y_{i+1}$.    
\item We have that 
\begin{align*} 
q' &\in R^{\not\sim}_{y_i,y_i}(q_0); & q &\in R^{\not\sim}_{y_i^{-1}x_{i+1},y_i^{-1}y_{i+1}}(q'); %\\
& p &\in R^{\not\sim}_{\varepsilon,y_i^{-1}x_{i+1}}(q); & r &\in R^{\not\sim}_{\varepsilon,y_i^{-1}y_{i+1}}(q).
\end{align*}
\item If $i > 0$, we also have that 
$p \in R^{\not\sim}_{\varepsilon,y_i^{-1}x_{i+1}}(r)$ and $r \in R^{\not\sim}_{\varepsilon,y_i^{-1}y_{i+1}}(r)$.
\end{enumerate} 
\end{lemma} 

\begin{proof} 
Assume first that $R$ is not monadic decomposable. By Proposition~\ref{prop:char},
there are infinite sequences $\{u_i\}_{i \geq 0}$, $\{\gamma_i\}_{i \geq 0}$, and 
$\{\delta_i\}_{i \geq 0}$ of words over $\Sigma$ and a quinary tuple
$\mathfrak{C}$ over ${\sf types}(R^{\not\sim})$ such that for each $i \geq 0$ it
is the case that 
\begin{enumerate} 
\item $|\gamma_i| = |\delta_i| > 0$, 
\item $u_i = \delta_0 \cdots \delta_{i-1} \gamma_i$,  
\item $(u_i,u_j) \in R^{\not\sim}$, for each $j > i$, and 
\item $\mathfrak{C}_{i,j} = \mathfrak{C}$, for each $j > i$. 
\end{enumerate} 
We then define a sequence $\{(x_i,y_i)\}_{i \geq 0}$ such that $x_i := u_i$, for each $i \geq 0$, and $y_i$ is the prefix of $x_{i+1} = u_{i+1}$ 
that has the same length as $x_i = u_i$, i.e., $y_i=\tau(x_{i+1},|x_i|)$. Hence, $y_i = \rho({u}_i) = \delta_0 \cdots \delta_i$. Clearly,
$|x_i| = |y_i| \geq 0$ and $y_i$ is a prefix of both $x_{i+1}$ and $y_{i+1}$,
for each $i \geq 0$. We prove next that the sequence $\{(x_i,y_i)\}_{i \geq 0}$
also satisfies the remaining conditions.   

Before defining $q,q',p,r \in Q$, let us highlight the intuition why such states
exist for every~$i$. We can find such states because by our assumption
$\mathfrak{C}_{i,j}=\mathfrak{C}$ for each $i < j$. Further, whether $q$ is reachable
from $q_0$ is stored in the first component of $\mathfrak{C}$. Similarly, the second
and third components of $\mathfrak{C}$ allow us to find $q'$ that is reachable
from $q_0$ and such that $q$ is reachable from $q'$. Finally, the fourth component
is for checking whether $p$ is reachable from $q$ and $r$, while the fifth component
for checking that $r$ is reachable from both $q$ and $r$.

Let us define $q,q',p,r \in Q$ as follows.
\begin{itemize} 
\item $q$ and $p$ are states such that $p \in F$ and it is the case that $q \in R^{\not\sim}_{x_0,y_0}(q_0)$ and $p \in R^{\not\sim}_{\varepsilon,y_0^{-1}x_1}(q)$. 
%\begin{align*}
%q \in R^{\not\sim}_{x_0,y_0}(q_0) \qquad \text{and} \qquad p \in R^{\not\sim}_{\varepsilon,y_0^{-1}x_1}(q).
%\end{align*} 
Notice that such $q$ and $p$ must exist as $(x_0,x_1) \in R^{\not\sim}$, i.e., it holds that 
$R^{\not\sim}_{x_0,x_1}(q_0) \cap F \neq \emptyset$, 
and $R^{\not\sim}_{x_0,x_1}(q_0) = R^{\not\sim}_{x_0,y_0}(q_0) \circ R^{\not\sim}_{\varepsilon,y_0^{-1}x_1}$. 
\item $q'$ is a state such that $q' \in R^{\not\sim}_{y_0,y_0}(q_0)$ and $q \in R^{\not\sim}_{y_0^{-1}x_1,y_0^{-1}y_1}(q')$. 
%\begin{align*}
%q' \in R^{\not\sim}_{y_0,y_0}(q_0) \qquad \text{and} \qquad  q &\in R^{\not\sim}_{y_0^{-1}x_1,y_0^{-1}y_1}(q').
%\end{align*} 
Notice that such a $q'$ must exist. Indeed, since $\mathfrak{C}_{0,1} = \mathfrak{C}_{1,2} = \mathfrak{C}$,  
we have $R^{\not\sim}_{u_0,\rho({u}_0)} = R^{\not\sim}_{x_0,y_0} = R^{\not\sim}_{u_1,\rho(u_1)} 
= R^{\not\sim}_{x_1,y_1}$. This implies that $q \in R^{\not\sim}_{x_1,y_1}(q_0)= R^{\not\sim}_{y_0,y_0}(q_0) \circ R^{\not\sim}_{y_0^{-1}x_1,y_0^{-1}y_1}$, as we know that $q \in R^{\not\sim}_{x_0,y_0}(q_0)$ and there must be an intermediate state $q'$ that is reached after
reading $(y_0,y_0)$.

\item We have that $r$ is a state such that 
%$r \in R^{\not\sim}_{\varepsilon,y_0^{-1}y_1}(q)$, 
\begin{align*}
r \in R^{\not\sim}_{\varepsilon,y_0^{-1}y_1}(q); \qquad  p \in R^{\not\sim}_{\varepsilon,y_1^{-1}x_2}(r); \qquad \text{and} \qquad r &\in R^{\not\sim}_{\varepsilon,y_1^{-1}y_2}(r).
\end{align*}
The existence of such state $r$ is not obvious. We prove this as 
Lemma~\ref{lem:rexists} after this proof.

\end{itemize} 
%

%\begin{eqnarray*} 
%R^{\not\sim}(q_0,w_0,v_0) = q \ \ & \ \ R^{\not\sim}(q_0,v_0,v_0) = q' \\
%R^{\not\sim}(q,\epsilon,w_{1}^{\uparrow v_{0}}) = a \ \ & \ \ R^{\not\sim}(q',\epsilon,v_{1}^{\uparrow v_{0}}) = a'. 
%%R^{\not\sim}(a',\epsilon,w_{i+1}^{\uparrow v_{i}}) = a \ \ & \ \ R^{\not\sim}(a',\epsilon,v_{i+1}^{\uparrow v_{i}}) = a'. 
%\end{eqnarray*} 
We now prove that $q,q',p,r$ satisfy all the requirements in the statement of the Lemma. By definition, $q\in R^{\not\sim}_{x_0,y_0}(q_0)$ and $p \in F$. 
We can then prove by induction that for each $i \geq 0$ it is the case that 
\begin{align*} 
q' \in R^{\not\sim}_{y_i,y_i}(q_0);  \quad q \in R^{\not\sim}_{y_i^{-1}x_{i+1},y_i^{-1}y_{i+1}}(q'); \quad %\\ \label{eq:ind1}
p \in R^{\not\sim}_{\varepsilon,y_i^{-1}x_{i+1}}(q); \quad  r \in R^{\not\sim}_{\varepsilon,y_i^{-1}y_{i+1}}(q); 
\end{align*} 
and, in addition, that for each $i > 0$ it is the case that  $p \in R^{\not\sim}_{\varepsilon,y_i^{-1}x_{i+1}}(r)$ and $r \in R^{\not\sim}_{\varepsilon,y_i^{-1}y_{i+1}}(r)$. The base case $i = 0$ holds by definition. 

We now prove by induction that for each $i \geq 0$ it is the case that 
\begin{align} \label{eq:ind}
q' \in R^{\not\sim}_{y_i,y_i}(q_0);  \quad q \in R^{\not\sim}_{y_i^{-1}x_{i+1},y_i^{-1}y_{i+1}}(q'); \quad %\\ \label{eq:ind1}
p \in R^{\not\sim}_{\varepsilon,y_i^{-1}x_{i+1}}(q); \quad  r \in R^{\not\sim}_{\varepsilon,y_i^{-1}y_{i+1}}(q); 
\end{align} 
and, in addition, that for each $i > 0$ it is the case that  
\begin{align} \label{eq:ind2} 
p \in R^{\not\sim}_{\varepsilon,y_i^{-1}x_{i+1}}(r) \quad \text{and} \quad r &\in R^{\not\sim}_{\varepsilon,y_i^{-1}y_{i+1}}(r).
\end{align} 

%We start with \eqref{eq:ind} and \eqref{eq:ind1}.  
%It is important to note first that $q \in R^{\not\sim}(q_0,w_i,v_i)$, for each $i \geq 0$. 
%This follows from the facts that $q \in R^{\not\sim}(q_o,w_0,v_0)$ and $R^{\not\sim}(q_0,w_i,v_i) = R^{\not\sim}(q_0,w_0,v_0)$. The latter holds
%as $\mathfrak{C}_{0,1} = \mathfrak{C}_{i,i+1} = \mathfrak{C}$, and thus $R^{\not\sim}(u_0,\hat{u}_0) = R^{\not\sim}(w_0,v_0) = R^{\not\sim}(u_i,\hat{u}_i) 
%= R^{\not\sim}(w_i,v_i)$. 
We start with \eqref{eq:ind}, which we prove by induction. 

\begin{itemize}
\item Base case $i = 0$. We have already proven it  when we showed the existence of
$q,q',p$ and $r$.
%We have, by definition, that 
%(i) $q' \in R^{\not\sim}_{y_0,y_0}(q_0)$, (ii) $p \in R^{\not\sim}_{\varepsilon,y_0^{-1}x_{1}}(q)$, (iii)  
%$r \in R^{\not\sim}_{\varepsilon,y_0^{-1}y_{1}}(q)$, and (iv) $q \in R^{\not\sim}_{y_0^{-1}x_{1},
%y_0^{-1}y_{1}}(q')$. 
%This follows from the fact that $R^{\not\sim}(q',w_{1}^{\uparrow v_{0}},
%v_{1}^{\uparrow v_{0}}) = R^{\not\sim}(q_0,w_1,v_1) = q$ since $R^{\not\sim}(q_0,v_0,v_0) = q'$. 
%But 
%$R^{\not\sim}(q_0,w_1,v_1) = R^{\not\sim}(q_0,w_0,v_0) = q$, as $R^{\not\sim}(w_0,v_0) = R^{\not\sim}(w_1,v_1)$. In fact, the latter holds
%since $\mathfrak{C}_{0,1} = \mathfrak{C}_{1,2} = \mathfrak{C}$, and thus $R^{\not\sim}(u_0,\hat{u}_0) = R^{\not\sim}(w_0,v_0) = R^{\not\sim}(u_1,\hat{u}_1) 
%= R^{\not\sim}(w_1,v_1)$.   

\item Inductive case $i+1$, for $i \geq 0$. 

\begin{itemize}

\item $q' \in R^{\not\sim}_{y_{i+1},y_{i+1}}(q_0)$. This is the case since $q' \in R^{\not\sim}_{y_0,y_0}(q_0)$ and 
$R^{\not\sim}_{y_{i+1},y_{i+1}}(q_0) = R^{\not\sim}_{y_0,y_0}(q_0)$. The latter holds 
since $\mathfrak{C}_{0,1} = \mathfrak{C}_{i+1,i+2} = \mathfrak{C}$, and thus $R^{\not\sim}_{\rho({u}_0),\rho({u}_0)} = R^{\not\sim}_{y_0,y_0} = R^{\not\sim}_{\rho({u}_{i+1}),\rho({u}_{i+1})} 
= R^{\not\sim}_{y_{i+1},y_{i+1}}$.

\item $q \in R^{\not\sim}_{y_i^{-1}x_{i+1},y_i^{-1}y_{i+1}}(q')$. This holds as we have that 
\begin{align*}
R^{\not\sim}_{x_{i+1},y_{i+1}}(q_0)  &=  R^{\not\sim}_{y_{i},y_{i}}(q_0) \circ R^{\not\sim}_{y_i^{-1}x_{i+1},y_i^{-1}y_{i+1}},
\end{align*}
$q' \in R^{\not\sim}_{y_{i},y_{i}}(q_0)$ by the previous item, and it is the case that $q \in R^{\not\sim}_{y_i^{-1}x_{i+1},y_i^{-1}y_{i+1}}(q')$. 
The latter is the case as 
$\mathfrak{C}_{0,1} = \mathfrak{C}_{i,i+1} = \mathfrak{C}$, and thus 
\begin{align*}
R^{\not\sim}_{\sigma(u_1,|u_0|),\sigma(\rho(u_1),|u_0|)} = R^{\not\sim}_{y_0^{-1}x_1,y_0^{-1}y_1} = R^{\not\sim}_{\sigma(u_{i+1},|{u}_i|),\sigma(\rho({u}_{i+1}),|u_i|)} 
= R^{\not\sim}_{y_i^{-1}x_{i+1},y_i^{-1}y_{i+1}}.
\end{align*}
This implies that $q \in R^{\not\sim}_{y_i^{-1}x_{i+1},y_i^{-1}y_{i+1}}(q')$ as 
$q \in R^{\not\sim}_{y_0^{-1}x_{1},y_0^{-1}y_{1}}(q')$. 

\item $p \in R^{\not\sim}_{\varepsilon,y_i^{-1}x_{i+1}}(q)$. This is the case since $
R^{\not\sim}_{\varepsilon,y_i^{-1}x_{i+1}}(q)= R^{\not\sim}_{\varepsilon,y_0^{-1}x_1}(q)$, which holds
since $\mathfrak{C}_{0,1} = \mathfrak{C}_{i,i+1} = \mathfrak{C}$, and thus 
\begin{align*}
R^{\not\sim}_{\varepsilon,\sigma(x_1,|x_0|)} \,= \, 
R^{\not\sim}_{\varepsilon,y_0^{-1}x_1} \, = \, R^{\not\sim}_{\varepsilon,\sigma(x_{i+1},|x_i|)} \, = \, R^{\not\sim}_{\varepsilon,y_i^{-1}x_{i+1}}.
\end{align*} 
Now the result follows from the fact that $p \in R^{\not\sim}_{\varepsilon,y_0^{-1}x_1}(q)$ by definition. 

\item $r \in R^{\not\sim}_{\varepsilon,y_i^{-1}y_{i+1}}(q)$. The proof is analogous to the previous case. 
\end{itemize}
This finishes the proof of \eqref{eq:ind}. 
\end{itemize}
We now prove \eqref{eq:ind2} by induction. 
\begin{itemize}
\item Base case $i = 1$. Again, we have by definition that 
$p \in R^{\not\sim}_{\varepsilon,y_1^{-1}x_2}(r)$ and $r \in R^{\not\sim}_{\varepsilon,y_1^{-1}y_2}(r)$. 

%First, notice that $R^{\not\sim}(q,\epsilon,w_1^{\uparrow v_{0}}) = R^{\not\sim}(q,\epsilon,w_2^{\uparrow v_{0}}) = a$ and 
%$R^{\not\sim}(q,\epsilon,v_1^{\uparrow v_{0}}) = R^{\not\sim}(q,\epsilon,v_2^{\uparrow v_{0}}) = a'$. This follows from the fact 
%that $\mathfrak{C}_{0,1} = \mathfrak{C}_{0,2} = \mathfrak{C}$, and thus $R^{\not\sim}(\epsilon,w_1^{\uparrow v_{0}}) = R^{\not\sim}(\epsilon,w_2^{\uparrow v_{0}})$ and 
%$R^{\not\sim}(\epsilon,v_1^{\uparrow v_{0}}) = R^{\not\sim}(\epsilon,v_2^{\uparrow v_{0}})$. 
%Moreover, $$R^{\not\sim}(q,\epsilon,w_{2}^{\uparrow v_{0}}) \, = \, R^{\not\sim}(R^{\not\sim}(q,\epsilon,v_1),\epsilon,w_2^{\uparrow v_1}) \, = \, R^{\not\sim}(a',\epsilon,w_2^{\uparrow v_1}),$$
%from which we conclude that $R^{\not\sim}(q,\epsilon,w_{2}^{\uparrow v_{0}}) = R^{\not\sim}(a',\epsilon,w_2^{\uparrow v_1}) = a$, and 
%$$R^{\not\sim}(q,\epsilon,v_{2}^{\uparrow v_{0}}) \, = \, R^{\not\sim}(R^{\not\sim}(q,\epsilon,v_1),\epsilon,v_2^{\uparrow v_1}) \, = \, R^{\not\sim}(a',\epsilon,v_2^{\uparrow v_1}),$$
%from which we conclude that $R^{\not\sim}(q,\epsilon,v_{2}^{\uparrow v_{0}}) = R^{\not\sim}(a',\epsilon,v_2^{\uparrow v_1}) = a'$. 

\item Inductive case $i+1$, for $i \geq 1$. We have that $R^{\not\sim}_{\varepsilon,y_{i+1}^{-1}x_{i+2}} = R^{\not\sim}_{\varepsilon,y_1^{-1}x_{2}}$ 
and $R^{\not\sim}_{\varepsilon,y_{i+1}^{-1}x_{i+2}} = R^{\not\sim}_{\varepsilon,y_1^{-1}x_{2}}$. This follows from the fact that $\mathfrak{C}_{0,1} = \mathfrak{C}_{i+1,i+2} = \mathfrak{C}$. Therefore, 
\begin{align*} 
p  \in R^{\not\sim}_{\varepsilon,y_{i+1}^{-1}x_{i+2}}(r)  & =  R^{\not\sim}_{\varepsilon,y_1^{-1}x_2}(r) \quad \text{ and } \quad
r  \in R^{\not\sim}_{\varepsilon,y_{i+1}^{-1}y_{i+2}}(r)  = R^{\not\sim}_{\varepsilon,y_1^{-1}y_2}(r).
\end{align*} 
\end{itemize}

This concludes the proof of the first direction.

Let us assume now that  there are an infinite sequence $\{(x_i,y_i)\}_{i \geq 0}$ of
pairs of words over $\Sigma$ and states $q,q',p,r \in Q$ that satisfy the
conditions stated in the statement of the lemma. We prove next that $R$ is not monadic decomposable 
by showing that there are infinite sequences $\{w_i\}_{i \geq 0}$, 
$\{\alpha_i\}_{i \geq 0}$ and
$\{\beta_i\}_{i \geq 0}$ of words over $\Sigma$ such that $\{w_i\}_{i \geq 0}$,
$\{\alpha_i\}_{i \geq 0}$, and $\{\beta_i\}_{i \geq 0}$ satisfy the conditions stated
in Lemma~\ref{lm:equivrel}. 

We define $w_i := x_i$ for each $i \geq 0$. Furthermore, 
$\alpha_0 := x_0$, $\beta_0 := y_0$, and for each $i > 0$ we set 
$\alpha_i := y_{i-1}^{-1}x_i$ and $\beta_i := y_{i-1}^{-1}y_i$. Clearly
$|\alpha_i| = |\beta_i| > 0$ 
and $w_i = x_i = \beta_0 \cdots \beta_{i-1} \alpha_i$, for each  $i \geq 0$. We prove next that $(w_i,w_j) \in R^{\not\sim}$
for each $0 \leq i < j$. Actually, we prove a stronger claim: 
$p \in R^{\not\sim}_{w_i,w_j}(q_0)$ and $r \in R^{\not\sim}_{w_i,\rho(w_j)}(q_0)$, for each $0 \leq i < j$, 
where as before $\rho(w_j) = \tau(w_{j+1},|w_j|) = \beta_0 \beta_1 \cdots \beta_j$. 
The result follows since $p \in F$ by assumption. 
is by induction on $j \geq 1$. 

\begin{itemize} 

\item Base case $j = 1$. We only have to consider $i = 0$. Then it is the case that 
\begin{align*}
p \, \in \, R^{\not\sim}_{\varepsilon,y_0^{-1}x_1}(q)  \, \subseteq \, R^{\not\sim}_{x_0,y_0}(q_0) \circ R^{\not\sim}_{\varepsilon,y_0^{-1}x_1} \, = \, R^{\not\sim}_{x_0,x_1}(q_0) \, = \, R^{\not\sim}_{w_0,w_1}(q_0),
\end{align*}
and, in addition, that 
\begin{align*}
r \, \in \, R^{\not\sim}_{\varepsilon,y_0^{-1}y_1}(q)  \, \subseteq \, R^{\not\sim}_{x_0,y_0}(q_0) \circ R^{\not\sim}_{\varepsilon,y_0^{-1}y_1} \, = \, R^{\not\sim}_{x_0,y_1}(q_0) \, = \, R^{\not\sim}_{w_0,\rho(w_1)}(q_0).
\end{align*}
%$$R^{\not\sim}(q_0,w_0,v_1) \, = \, R^{\not\sim}(q_0,w_0,v_0) \circ R^{\not\sim}(\epsilon,v_1^{\uparrow v_0}) \, = \, R^{\not\sim}(q,\epsilon,v_1^{\uparrow v_0}) \, = \, a'.$$
%Recall that $a \in F$ by assumption.  

\item Inductive case $j+1$, for $j \geq 1$. Consider an arbitrary $i$ with $0 \leq i < j$. %If $i < j$ 
We have that 
\begin{align*}
p \, \in \, R^{\not\sim}_{\varepsilon,y_j^{-1}x_{j+1}}(r)  \, \subseteq \, R^{\not\sim}_{x_i,y_j}(q_0) \circ R^{\not\sim}_{\varepsilon,y_j^{-1}x_{j+1}} \, = \, R^{\not\sim}_{x_i,x_{j+1}}(q_0) \, = \, R^{\not\sim}_{w_i,w_{j+1}}(q_0),
\end{align*}
%
%$$R^{\not\sim}(q_0,w_i,w_{j+1}) \, = \, R^{\not\sim}(R^{\not\sim}(q_0,w_i,v_j),\epsilon,w_{j+1}^{\uparrow v_j}) \, = \, R^{\not\sim}(a',\epsilon,w_{j+1}^{\uparrow v_j}) \, = \, a,$$
where the containment holds by induction hypothesis. 
%
%second equality holds by hypothesis induction and the last one by assumption. 
Analogously, we have that 
\begin{align*}
r \, \in \, R^{\not\sim}_{\varepsilon,y_j^{-1}y_{j+1}}(r)  \, \subseteq \, R^{\not\sim}_{x_i,y_j}(q_0) \circ R^{\not\sim}_{\varepsilon,y_j^{-1}y_{j+1}} \, = \, R^{\not\sim}_{x_i,y_{j+1}}(q_0) \, = \, R^{\not\sim}_{w_i,\rho(w_{j+1})}(q_0).
\end{align*}
%
%$$R^{\not\sim}(q_0,w_i,v_{j+1}) \, = \, R^{\not\sim}(R^{\not\sim}(q_0,w_i,v_j),\epsilon,v_{j+1}^{\uparrow v_j}) \, = \, R^{\not\sim}(a',\epsilon,v_{j+1}^{\uparrow v_j}) \, = \, a'.$$

\iffalse
On the other hand, if $i = j$  we have that 
\begin{multline*} 
a \in  R^{\not\sim}(q,\epsilon,w_{j+1}^{\uparrow v_j}) \, \subseteq \,   \\ 
\big(R^{\not\sim}(q_0,v_{j-1},v_{j-1}) \circ R^{\not\sim}(w_j^{\uparrow v_{j-1}},v_j^{\uparrow v_{j-1}})\big) \circ R^{\not\sim}(\epsilon,w_{j+1}^{\uparrow v_j})  \, = \, \\  
 R^{\not\sim}(q_0,w_j,w_{j+1}), 
\end{multline*} 
%\begin{multline*} 
%R^{\not\sim}(q_0,w_j,w_{j+1}) \, = \, R^{\not\sim}(R^{\not\sim}(q_0,w_j,v_j),\epsilon,w_{j+1}^{\uparrow v_j}) \, = \\ 
%R^{\not\sim}(R^{\not\sim}(R^{\not\sim}(q_0,v_{j-1},v_{j-1}),w_j^{\uparrow v_{j-1}},v_j^{\uparrow v_{j-1}}),\epsilon,w_{j+1}^{\uparrow v_j}) \, = \\ R^{\not\sim}(R^{\not\sim}(q',w_j^{\uparrow v_{j-1}},v_j^{\uparrow v_{j-1}}),\epsilon,w_{j+1}^{\uparrow v_j})
% \, = \,  R^{\not\sim}(q,\epsilon,w_{j+1}^{\uparrow v_j}) \, = \, a,\end{multline*}
and, analogously, that 
\begin{multline*} 
a' \in  R^{\not\sim}(q,\epsilon,v_{j+1}^{\uparrow v_j}) \, \subseteq \,   \\ 
\big(R^{\not\sim}(q_0,v_{j-1},v_{j-1}) \circ R^{\not\sim}(w_j^{\uparrow v_{j-1}},v_j^{\uparrow v_{j-1}})\big) \circ R^{\not\sim}(\epsilon,v_{j+1}^{\uparrow v_j})  \, = \, \\  
 R^{\not\sim}(q_0,w_j,v_{j+1}).  
\end{multline*} \fi
\end{itemize} 
 
 This finishes the proof of the theorem. 
\end{proof} 

Next we prove the existence of state $r$ satisfying
\begin{align*}
r \in R^{\not\sim}_{\varepsilon,y_0^{-1}y_1}(q); \qquad  p \in R^{\not\sim}_{\varepsilon,y_1^{-1}x_2}(r); \qquad \text{and} \qquad r &\in R^{\not\sim}_{\varepsilon,y_1^{-1}y_2}(r).
\end{align*}

\begin{lemma}\label{lem:rexists}
Let $\{(x_i,y_i)\}_{i\geq0}$, $q\in Q$ and $p\in F$ and $\mathfrak{C}\in{\sf types}(R^{\not\sim})^5$ as defined in the proof of Lemma~\ref{theo:main}. 
There is a state $r \in Q$ such that 
\begin{align*}
r &\in R^{\not\sim}_{\varepsilon,y_0^{-1}y_1}(q); & p &\in R^{\not\sim}_{\varepsilon,y_1^{-1}x_2}(r); & r &\in R^{\not\sim}_{\varepsilon,y_1^{-1}y_2}(r).
\end{align*}
\end{lemma} 
\begin{proof} 
First of all, let $r_1, r_2, \ldots$ be an infinite sequence of states in $Q$ that satisfies the following. 
\begin{itemize}
\item
We have that $r_1$ is any state that 
satisfies $r_1 \in R^{\not\sim}_{\varepsilon,y_0^{-1}y_1}(q)$ and $p \in R^{\not\sim}_{\varepsilon,y_1^{-1}x_2}(r_1)$. 

\item For each $j > 1$ it is the case that
$r_{j} \in R^{\not\sim}_{\varepsilon,y_{j-1}^{-1}y_{j}}(r_{j-1})$ and $p \in R^{\not\sim}_{\varepsilon,y_j^{-1}x_{j+1}}(r_{j})$.

\end{itemize} 

We explain next why the sequence $r_1, r_2, \ldots$ is well-defined. 
Notice first that state $r_1$ must exist 
as $p \in R^{\not\sim}_{\varepsilon,y_0^{-1}x_2}(q)$ and 
\begin{align*}
R^{\not\sim}_{\varepsilon,y_0^{-1}x_2}(q) = R^{\not\sim}_{\varepsilon,y_0^{-1}y_1}(q) \circ R^{\not\sim}_{\varepsilon,y_1^{-1}x_2}.
\end{align*}
The former is the case as $\mathfrak{C}_{0,1} = \mathfrak{C}_{0,2} = \mathfrak{C}$, and thus
$R^{\not\sim}_{\varepsilon,y_0^{-1}x_1} = R^{\not\sim}_{\varepsilon,y_0^{-1}x_2}$. 
This implies that $p \in R^{\not\sim}_{\varepsilon,y_0^{-1}x_2}(q)$, as we know that $p \in R^{\not\sim}_{\varepsilon,y_0^{-1}x_1}(q)$.

Assume now that we have identified states $r_1, r_2, \ldots, r_k$, for $k \geq 1$, such that 
for each $j \leq k$ it is the case that
$r_{j} \in R^{\not\sim}_{\varepsilon,y_{j-1}^{-1}y_{j}}(r_{j-1})$ and $p \in R^{\not\sim}_{\varepsilon,y_j^{-1}x_{j+1}}(r_{j})$. 
Then $r_{k+1}$ is any state in $Q$ that 
satisfies $r_{k+1} \in R^{\not\sim}_{\varepsilon,y_k^{-1}y_{k+1}}(r_{k})$ and $p \in R^{\not\sim}_{\varepsilon,y_{k+1}^{-1}x_{k+2}}(r_{k+1})$. 
Notice that state $r_{k+1}$ must exist 
as $p \in R^{\not\sim}_{\varepsilon,y_k^{-1}x_{k+2}}(r_k)$ and 
\begin{align*}
R^{\not\sim}_{\varepsilon,y_k^{-1}x_{k+2}}(r_k) \, = \, R^{\not\sim}_{\varepsilon,y_k^{-1}x_{k+1}}(r_k) \circ 
R^{\not\sim}_{\varepsilon,y_{k+1}^{-1}x_{k+2}}.
\end{align*}
Again, the former is the case as $\mathfrak{C}_{k,k+1} = \mathfrak{C}_{k,k+2} = \mathfrak{C}$, and thus
$R^{\not\sim}_{\varepsilon,y_k^{-1}x_{k+1}} = R^{\not\sim}_{\varepsilon,y_k^{-1}x_{k+2}}$. 
This implies that $p \in R^{\not\sim}_{\varepsilon,y_k^{-1}x_{k+2}}(r_k)$, as we know that 
$p \in R^{\not\sim}_{\varepsilon,y_k^{-1}x_{k+1}}(r_k)$ by hypothesis.

An important property of the sequence $r_1, r_2, \ldots$, as defined above, is that $r_k \in R^{\not\sim}_{\varepsilon,y_0^{-1}y_k}(q)$, for each $k \geq 1$, and 
$r_k \in R^{\not\sim}_{\varepsilon,y_j^{-1}y_k}(r_j)$, for each $1 \leq j < k$. This can be proved easily by induction on $k \geq 1$. 

Since the sequence $r_1, r_2, \ldots$ is infinite, there must be integers 
$1 \leq j < k$ such that $r_j = r_k$. Therefore, it is the case that 
\begin{enumerate} 
\item $r_j \in R^{\not\sim}_{\varepsilon,y_0^{-1}y_j}(q)$; 
\item $p \in R^{\not\sim}_{\varepsilon,y_j^{-1}x_{j+1}}(r_j)$, which holds by definition of $r_j$ in the sequence $r_1, r_2, \ldots$; 
\item $r_k \in R^{\not\sim}_{\varepsilon,y_j^{-1}y_k}(r_j)$, and thus 
$r_j \in R^{\not\sim}_{\varepsilon,y_j^{-1}y_k}(r_j)$; and 
\item $p \in R^{\not\sim}_{\varepsilon,y_j^{-1}x_{k+1}}(r_j)$. 
\end{enumerate} 
It follows then that 
\begin{align*}
r_j &\in R^{\not\sim}_{\varepsilon,y_0^{-1}y_1}(q); & p &\in R^{\not\sim}_{\varepsilon,y_1^{-1}x_2}(r_j); & r_j &\in R^{\not\sim}_{\varepsilon,y_1^{-1}y_2}(r_j).
\end{align*}
In fact: 
\begin{itemize}
\item $r_j \in R^{\not\sim}_{\varepsilon,y_0^{-1}y_1}(q)$, as $r_j \in R^{\not\sim}_{\varepsilon,y_0^{-1}y_j}(q)$ by hypothesis and, in addition, we have that 
$\mathfrak{C}_{0,1} = \mathfrak{C}_{0,j} = \mathfrak{C}$, i.e., $R^{\not\sim}_{\varepsilon,y_0^{-1}y_{j}} = R^{\not\sim}_{\varepsilon,y_0^{-1}y_{1}}$. 
%This implies that $a_j \in R^{\not\sim}(q,\epsilon,v_1^{\uparrow v_0})$, as we know that 
%$a_j \in R^{\not\sim}(q,\epsilon,v_j^{\uparrow v_0})$ by hypothesis.
\item $p \in R^{\not\sim}_{\varepsilon,y_1^{-1}x_2}(r_j)$, as $p \in R^{\not\sim}_{\varepsilon,y_j^{-1}x_{j+1}}(r_j)$ by hypothesis and, in addition, we have that 
$\mathfrak{C}_{j,j+1} = \mathfrak{C}_{1,2} = \mathfrak{C}$, i.e., $R^{\not\sim}_{\varepsilon,y_j^{-1}x_{j+1}} = R^{\not\sim}_{\varepsilon,y_1^{-1}x_{2}}$. 
\item $r_j \in R^{\not\sim}_{\varepsilon,y_1^{-1}y_2}(r_j)$, as $r_j \in R^{\not\sim}_{\varepsilon,y_j^{-1}y_{k}}(r_j)$ by hypothesis and, in addition, we have that 
$\mathfrak{C}_{j,k} = \mathfrak{C}_{1,2} = \mathfrak{C}$, i.e., $R^{\not\sim}_{\varepsilon,y_j^{-1}y_{k}} = R^{\not\sim}_{\varepsilon,y_1^{-1}y_{2}}$. 
\end{itemize} 
We can then set $r = r_j$. 
This finishes the proof of the claim. 
\end{proof}

The runs as extracted from the sequence $\{(x_i,y_i\})_{i\geq0}$ satisfying the conditions of Lemma~\ref{theo:main} are depicted in Figure~\ref{fig:runs}.

\begin{figure}
\centering
\begin{tikzpicture}[initial text=, initial distance=1.5ex,>=stealth, every state/.style={inner sep=2pt, minimum size=0pt},decoration=snake,line around/.style={decoration={pre length=#1, post length=#1}}]
\node[state,initial] (q0) {$q_0$};
\node[state] (q') [right = 1.5cm of q0] {$q'$};
\node[state] (q) [right = 3.5cm of q'] {$q$};
\node[state,accepting] (p) [right = 2.5cm of q] {$p$};
\node[state] (r) [below right = 2cm and 1.25cm of q] {$r$}; 

\path[->] (q0) edge[decorate,line around=2pt] node[above] {$(y_i,y_i)$} (q')
		  (q') edge[decorate,line around=2pt] node[above] {$(y_i^{-1}x_{i+1},y_i^{-1}y_{i+1})$} (q)
		  (q) edge[decorate,line around=2pt] node[above] {$(\varepsilon,y_i^{-1}x_{i+1})$} (p)
		  (q) edge[decorate,line around=2pt] node[below,sloped] {$(\varepsilon,y_i^{-1}y_{i+1})$} (r)
		  (r) edge[decorate,line around=2pt] node[above,sloped] {$(\varepsilon,y_{i+1}^{-1}x_{i+2})$} (p)
		  (r) edge[decorate,line around=2pt,in=-55, out=30, loop] node[right,xshift=0.1cm] {$(\varepsilon,y_{i+1}^{-1}y_{i+2})$} (q');
\end{tikzpicture}
\caption{\label{fig:runs}Runs in $R^{\not\sim}$ on states $q$, $q'$, $p$ and $r$ as defined in Lemma~\ref{theo:main}. The runs exist for every $i\geq0$.}
\end{figure}

Let us briefly return to Example~\ref{ex:stairs}. Now using the notation of Lemma~\ref{theo:main}, the states $q'$, $q$, $r$ and $p$ are
$q'=q_0$, $q=q_1$ and $r=p=q_3$. Observe that a run from $q_0$ to $p$ through $r$
(i.e., a run from $q_0$ to $q_3$) can be used to construct an infinite sequence
of runs on $(w,w')\in R^{\not\sim}$.

Lemma~\ref{theo:main} allows us to reduce the monadic decomposability problem to a set of reachability checks on types.  
With the help of this property, we can then prove Lemma~\ref{theo:ptime}. 

\begin{proof}[Proof of Lemma~\ref{theo:ptime}] 
For each $(q,q',p,r) \in Q \times Q \times Q \times Q$ with $p \in F$ do the following. 
\begin{itemize} 
\item Check if there are words $w_0, v_0,w_1,v_1$ such that $|w_0| = |v_0| > 0$, $|w_1| = |v_1| > 0$, and it holds that 
(i) $q \in R^{\not\sim}_{w_0,v_0}(q_0)$, (ii) $q'\in  R^{\not\sim}_{v_0,v_0}(q_0)$, (iii) $q \in 
R^{\not\sim}_{w_1,v_1}(q')$, (iv) $q' \in R^{\not\sim}_{v_1,v_1}(q')$, (v) 
$p \in R^{\not\sim}_{\varepsilon,w_1}(q)$, and 
(vi) $r \in R^{\not\sim}_{\varepsilon,v_1}(q)$. 
\item Check if there are words $w,v$ such that $|w| = |v| > 0$, and it holds that 
(i) $q \in R^{\not\sim}_{w,v}(q')$, (ii) $q' \in R^{\not\sim}_{v,v}(q')$, (iii) 
$p \in R^{\not\sim}_{\varepsilon,w}(q)$,  
(vi) $r \in R^{\not\sim}_{\varepsilon,v}(q)$, (v) $p \in R^{\not\sim}_{\varepsilon,w}(r)$, and (vi) $r \in R^{\not\sim}_{\varepsilon,v}(r)$.  
\end{itemize} 
If this holds for any such a tuple, then $R$ is not monadic decomposable. Else, $R$ is monadic decomposable. 
It is easy to see that this algorithm can be implemented in $\NLOGSPACE$.  
\end{proof}

We have the necessary ingredients to prove a part of Theorem~\ref{thm:main}.
\begin{lemma}\label{lem:part1}
Deciding whether a given binary regular relation $R$ is monadic decomposable is 
    in $\NLOGSPACE$ (resp.~in $\PSPACE$), if $R$ is given by a 
    DFA (resp.~an NFA). 
\end{lemma}
\begin{proof}
The claim follows from Lemma~\ref{theo:ptime}. Namely, from the definition
of $R^{\not\sim}$, it follows that, if $R$ is given by a DFA, then $R^{\not\sim}$
can be constructed in $\LOGSPACE$.
%and hence monadic decomposability is decidable in $\NLOGSPACE$.
%\pablo{We need something else: that $R^{\not\sim}$ can actually be constructed in logspace}
Indeed, this can be done as disjunctions, conjunctions and projections can all be
done in $\LOGSPACE$ and then via composability of $\LOGSPACE$ transducers
we can construct $R^{\not\sim}$ of logarithmic size. (Note that the output of a 
$\LOGSPACE$ transducer is of at most polynomial size.) Then by Lemma~\ref{theo:ptime},
we obtain the decidability of monadic decomposability in $\NLOGSPACE$ for $R$
given by a DFA.

Similarly, if $R$ is given by an NFA, we construct $R^{\not\sim}$ of polynomial
size since an NFA can be transformed into a DFA using a $\PSPACE$ transducer. 
(Again, the output of a $\PSPACE$ transducer is of at most exponential size.)
Thus monadic decomposability is in $\PSPACE$.
\end{proof}

\section{Deciding monadic decomposability of regular relations}
\label{sec:general}
In this section, we finish the proof of Theorem~\ref{thm:main}. The remaining
component is showing that monadic decomposability of $n$-ary regular relations
is decidable in $\NLOGSPACE$ for DFA and $\PSPACE$ for NFA.

\begin{lemma}\label{thm:nary}
%    There is an algorithm that
%    takes as input a DFA for an $n$-ary regular relation $R$
%    and decides if $R$ is monadic decomposable in $\NLOGSPACE$.
%	If the input is an NFA for an $n$-ary regular relation $R$, then
%	the algorithm decides if $R$ is monadic decomposable in $\PSPACE$.
Deciding whether a given $n$-ary regular relation $R$ is monadic decomposable is 
    in $\NLOGSPACE$ (resp.~in $\PSPACE$), if $R$ is given by a 
    DFA (resp.~an NFA). 
\end{lemma}

\begin{proof}[Proof of Theorem~\ref{thm:main}]
The upper bounds follow from Lemma~\ref{thm:nary} and the lower bound follows
from Lemma~\ref{thm:lb} for NFA and from Lemma~\ref{lem:lbdfa} for DFA.
\end{proof}

In order to prove Lemma~\ref{thm:nary}, we extend Lemma~\ref{lem:part1} to $n$-ary relations. Let us first define
some helpful notation used throughout the section.

Recall that words of regular relations are padded to be of the same length using
$\pad$. We denote this function by $\PADD{\pad}$. For example,
$\PAD{\pad}{(a,\varepsilon,ab)} = (a\pad, \pad\pad, ab)$.
Let us now define a padding function $\delta_n$ that acts slightly differently.
Instead of padding the words in a tuple to make them of the same length, the new
function pads a sequence of tuples with tuples where some elements are $\pad$.
Let us describe $\delta_n$ in more details. Define
$\Sigma_n = (\Sigma_\pad)^n\setminus\{\pad^n\}$, i.e., an alphabet consisting of
$n$-tuples of letters from $\Sigma_\pad$, excluding $(\pad, \ldots, \pad)$. Now $\delta_n : (\Sigma^*)^n \to \Sigma_n^*$ is an injective mapping that uses $\pad$ to extend
the shorter words to the same length as the longest word.
For example, $\delta_3$ maps $(a, \varepsilon, ab) \in (\Sigma^*)^3$ to
$(a,\pad,a) (\pad,\pad,b) \in \Sigma_3^*$ as follows:
\begin{align*}
(a,\varepsilon,ab) \xrightarrow{}
\begin{pmatrix}
a\\
\varepsilon\\
ab
\end{pmatrix} \xrightarrow{}
\begin{pmatrix}
a\pad\\
\pad\pad\\
a b
\end{pmatrix}  \xrightarrow{}
\begin{pmatrix}
a\\
\bot\\
a
\end{pmatrix}\begin{pmatrix}
\bot\\
\bot\\
b
\end{pmatrix} \xrightarrow{}
(a,\pad,a)(\pad,\pad,b).
\end{align*}

\begin{lemma}
    For $n\ge 1$, $\{(x_1, \ldots, x_n, y) \mid
    \delta_n(x_1, \ldots, x_n) = y \}
    \subseteq (\Sigma^*)^n \times \Sigma_n^* $ is regular.
\end{lemma}

Given an $n$-ary relation $R \subseteq (\Sigma^*)^n$
and positive integers $x_1, \dots, x_m$ such that $\sum_{i=1}^m x_i = n$,
an $m$-ary relation $R_{x_1, \dots, x_m} \subseteq
\Sigma_{x_1}^* \times \cdots \times \Sigma_{x_m}^*$
can be uniquely determined via the mappings $\delta_{x_1}, \dots, \delta_{x_m}$.
More precisely, there exists a one-to-one correspondence $\Delta_{x_1, \dots, x_m}$
between relations $R$ and $R_{x_1, \dots, x_m}$ that maps each
$(w_1, \ldots, w_n) \in R$ to% \qquad \xLeftrightarrow{\Delta_{a_1, \dots, a_m}}\\
\begin{align*}
  (\delta_{x_1}(w_1, \ldots, w_{x_1}), \delta_{x_2}(w_{x_1+1}, \ldots, w_{x_1+x_2}),
    \ldots, \delta_{x_m}(w_{x_1+\cdots+ x_{m-1}+1}, \ldots, w_n))
\in R_{x_1, \dots, x_m}.
\end{align*}
For example, a ternary relation $R=\{(a,\varepsilon, ab)\}$ over $(\Sigma^*)^3$
uniquely determines a binary relation $R_{1,2} = \{(a, (\pad,a)(\pad,b))\}$
over $\Sigma_1^*\times \Sigma_2^*$ through the correspondence $\Delta_{1,2}$.
For the sake of readability, if the integers $x_1,\ldots,x_m$ have a constant
subsequence of length $k$, i.e., $x_i=x_{i+1}=\cdots = x_{i+k-1}$ for some $i$, we write the relation
as $R_{x_1,\ldots,x_{i-1},x_i^{k},x_{i+k},\ldots,x_m}$.

In the following, we shall use $R_k$ to denote the binary relation
$R_{k,n-k}$ induced by $R$.
It turns out that being able to check monadic decomposability for binary relations
is sufficient to check monadic decomposability for general $n$-ary relations.
\begin{lemma}\label{lemma:nary}
Let $R$ be an $n$-ary regular relation and let $R_1, \dots , R_{n-1}$ be the induced binary relations. Then $R$ is monadic decomposable
iff
$R_1, \dots , R_{n-1}$ are monadic decomposable.
\end{lemma}
\begin{proof} 
Define $\delta_i(S) = \{ \delta_i(s_1, \ldots, s_i) \mid (s_1, \ldots, s_i)\in S \}$.
The only-if part of the lemma is immediate, since
$R = \bigcup\nolimits_i X_{i,1} \times \cdots \times X_{i,n}$
implies that
$R_k = \bigcup\nolimits_i \delta_k(X_{i,1} \times \cdots \times X_{i,k}) \times \delta_{n-k}(X_{i,k+1}\times \cdots \times X_{i,n})$
for $1 \le k \le n-1$, namely, $R_1, \dots , R_{n-1}$ are monadic decomposable.

To see the other direction,
we say that an $n$-ary relation $R$ is \emph{$k$-decomposable} if the induced $k$-ary relation
$R_{1^{k-1},n-k+1}$ of $R$ is monadic decomposable.
%where $X_{i,1},\dots, X_{i,k-1} \subseteq \Sigma^*$ and $Y \subseteq \Sigma_{n-k+1}^*$.
Now it suffices to show that $R$ is $n$-decomposable since $R = R_{1^n}$.
We shall prove this by induction on $k \in \{2,\ldots, n\}$.
Note that $R$ is 2-decomposable by the assumption that $R_1$ is monadic decomposable.
For $2\le k \le n-1$, suppose that $R_k = \bigcup\nolimits_j A_j \times B_j$
and $R$ is $k$-decomposable, say
%\begin{align*}
$R_{1^{k-1},n-k+1} = \bigcup\nolimits_i X_{i,1} \times \cdots \times X_{i,k-1} \times Y_i$.
%\end{align*}
Then $R$ is $(k+1)$-decomposable as we have
\begin{align*}
R_{1^k,n-k} =
\bigcup\nolimits_i \bigcup\nolimits_j X_{i,1} \times \cdots \times X_{i,k-1} \times A_{i,j} \times B_j,
\end{align*}
where
$A_{i,j} = \{ x \in \Sigma^* \mid \exists x_1\in X_{i,1}
\cdots \exists x_{k-1} \in X_{i,k-1}.~ \delta_k(x_1, \dots, x_{k-1}, x) \in A_j \}$,
i.e., $A_{i,j}$ is the projection of
$\delta_k^{-1}(A_j) \cap (X_{i,1} \times \cdots \times X_{i,k-1} \times \Sigma^*)$
on the $k$-th component.
Note that
$\delta_k^{-1}(A_j) $ is regular since $A_j$ and
$\{(x_1, \ldots, x_k, y) \mid \delta_k(x_1, \ldots, x_k) = y \}$
are regular (cf.~\cite{Bl99}).
Hence $A_{i,j}$ is also regular.
The claim that $R$ is $n$-decomposable then follows by induction.
\end{proof}

We can then obtain our desired result. 

\begin{proof}[Proof of Lemma~\ref{thm:nary}]
To prove the lemma, we show that if $R$ is regular, then so are the induced relations
$R_1,\ldots,R_{n-1}$. Moreover, given the automaton of $R$, one can construct the
automaton for each $R_i$ 
%for each such a relation can be constructed 
in
logarithmic space from $R$. 

We first show that if an $n$-ary relation $R$ is regular,
so are the induced binary relations $R_1, \ldots, R_{n-1}$.
Let $\A_\pad = (\Sigma_\pad, Q, \rightarrow_{\A}, q_0, F)$
be an $n$-tape automaton recognizing $R_\pad$ and fix $k \in \{1,\ldots, n-1\}$.
We argue that there exists a two-tape automaton
$\B_{\pad'} = (\Sigma'_{\pad'}, Q, \rightarrow_{\B}, q_0, F)$
recognizing $(R_k)_{\pad'}$.
The definitions of $\B_{\pad'}$ and $\A_\pad$ differ only
in their alphabet, padding symbols, and transition relations:
the alphabet of $\B_{\pad'}$ is
$\Sigma' = \Sigma_k \cup \Sigma_{n-k}$;
the padding symbol $\pad'$ is a fresh symbol
not used in $\Sigma \cup \Sigma'$;
the transition relation of $\B_{\pad'}$ is
determined from that of $\A_\pad$ in the following way:
for each transition
$\tau = (q, s_1, \ldots, s_n, S)$ in $\rightarrow_{\A}$,
there is a transition $\tau'$ in $\rightarrow_{\B}$ such that
$\tau' = (q, \lambda_k(s_1, \ldots, s_k), \lambda_{n-k}(s_{k+1}, \ldots, s_n), S)$
and vice versa.
Here $\lambda_i : (\Sigma_\pad)^i \to (\Sigma_i)_{\pad'}$
is defined by
\begin{align*}
\lambda_{i}(s_{1},\ldots,s_{i}) & =\begin{cases}
(s_{1}, \ldots, s_{i}), & \mbox{if \ensuremath{s_{j} \neq \bot} for some \ensuremath{j \in \{1,\ldots, i\}};}\\
\bot', & \mbox{otherwise}.
\end{cases}
\end{align*}
Now we show that $\B_{\pad'}$ recognizes $(R_k)_{\pad'}$.
It is easy to see that $\lambda_{k}$ and $\lambda_{n-k}$ together
induce a one-to-one mapping from $L(\A_\pad)$ to $L(\B_{\pad'})$.
Denote this mapping with $\Lambda_{k,n-k}$.
%Note that transitions in $\rightarrow_{\A}$ and $\rightarrow_{\B}$ only differ in their transition labels.
Given an arbitrary $w = (w_1, \ldots, w_n) \in R$,
where $w_i=a_{i,1}\cdots a_{i,m}$ for each $i\in\{1,\ldots,n\}$,
$n$-tape automaton $\A_\pad$ has a run, say,
\begin{align*}
q_{0}\xrightarrow{\mbox{\footnotesize$(a_{1,1},\ldots,a_{n,1})$}}
q_{1}\to\cdots\to q_{m-1}\xrightarrow{\mbox{\footnotesize$(a_{1,m},\ldots,a_{n,m})$}}q_{m},
\end{align*}
that accepts $\PAD{\pad}{w}$.
We will write $(u,v)$ as $\sdomi{u}{v}$ for the sake of readability when talking
about a run on $\B_{\pad'}$.
Now by definition, $\B_{\pad'}$ has a run
\begin{align*}
q_{0}\xrightarrow{\begin{pmatrix}
    \lambda_{k}(a_{1,1},\ldots,a_{k,1})\\
    \lambda_{n-k}(a_{k+1,1},\ldots,a_{n,1})
    \end{pmatrix}}q_{1}\to\cdots\to q_{m-1}\xrightarrow{\begin{pmatrix}
    \lambda_{k}(a_{1,m},\ldots,a_{k,m})\\
    \lambda_{n-k}(a_{k+1,m},\ldots,a_{n,m})
    \end{pmatrix}}q_{m}
\end{align*}
that accepts $\Lambda_{k,n-k}(\PAD{\pad}{w})$,
which can be written as
$\PAD{\pad'}{(b_1\cdots b_\ell,~c_1\cdots c_r)}$ with
$b_i = (a_{1,i},\dots,a_{k,i}) \in \Sigma_k$
for $i \in \{1, \ldots, \ell\}$,
$c_i = (a_{k+1,i},\ldots,a_{n,i}) \in \Sigma_{n-k}$
for $i \in \{1, \dots, r\}$, and $\max\{\ell, r\} = m$.
Note that
$b_1 \cdots b_\ell = \delta_{k}(w_1, \ldots, w_{k})$ and
$c_1 \cdots c_r = \delta_{n-k}(w_{k+1}, \ldots, w_{n})$.
Hence we have
$(b_1 \cdots b_\ell,~c_1 \cdots c_r) =
\Delta_{k,n-k}((w_1, \ldots,w_n)) = \Delta_{k,n-k}(w)$,
which implies that
$\Lambda_{k,n-k}(\PAD{\pad}{w}) = \PAD{\pad'}{\Delta_{k,n-k}(w)}$.
Since $w \in R$ is arbitrary,
the two mappings
$\Lambda_{k,n-k} \circ \PADD{\pad}$ and
$\PADD{\pad'} \circ \Delta_{k,n-k}$
coincide on domain $R$.
However, note that
$\Lambda_{k,n-k} \circ \PADD{\pad}$ and
$\PADD{\pad'} \circ \Delta_{k,n-k}$
are isomorphisms from $R$ to $L(\B_{\pad'})$
and from $R$ to $(R_k)_{\pad'}$, respectively.
It then follows that $L(\B_{\pad'}) = (R_k)_{\pad'}$.

In order to construct the $n-1$ automata for $R_i$'s (i.e., $\B_{\pad'}$)
with a logarithmic space transducer, observe that each transition in each automaton
$\B_{\pad'}$ is simply a projection of some transition in $R$, and hence the number
of transitions in $\B_{\pad'}$ is at most the number of transitions in $\A_\pad$.
Then, the logarithmic space transducer would need to keep track of which transition
in $\A_\pad$, which letter in $\Sigma_\pad$, and finally which position in the
product label $(a_1,\ldots,a_n)$ is being transformed. This can be done in 
$O(\log|\A_\pad| + \log|\Sigma_\pad| + \log n) = O( \log|\A_\pad| )$ space.
To show the $\NLOGSPACE$ (resp.~$\PSPACE$) complexity, we invoke the algorithm
of Lemma~\ref{lem:part1} for each $R_i$.
\end{proof}

\section{Concluding Remarks}
\label{sec:conc}

Monadic decomposability for rational relations (and subclasses thereof)
%that are represented by multi-tape
%automata 
is a classical problem in automata theory that dates back to the late
1960s (the work of Stearns \cite{Stearns} and Fischer and Rosenberg
\cite{FR68}). While the general problem is undecidable, the subcase of regular
relations (i.e. those recognized by synchronized multi-tape automata) provides
a good balance between decidability \cite{Lib00,CCG06} and expressiveness. 
The complexity of this subcase remained open for over a decade
(exponential-time upper bound
for the binary case \cite{LS17,LS19}, double exponential-time upper bound in
the general case \cite{CCG06}, and no specific lower bounds). This paper closes 
this question by providing the precise complexity for the problem: 
$\NLOGSPACE$ (resp.~$\PSPACE$) for DFA (resp.~NFA) representations. 

\smallskip
\noindent
\textbf{Some perspectives from formal verification and future work: } 
Researchers from the area of formal verification have increasingly understood
the importance of the monadic decompositions techniques, e.g., see
\cite{VBN+17}. Directly pertinent to monadic decomposability of regular 
relations is the line of work of constraint solving over strings, wherein
increasingly more complex string operations are needed and thus added to
solvers \cite{S3,Abdulla14,LB16,Abdulla17,CCHLW18,trau18,popl19}.
As an example, let us take a look at the recent work of Chen \emph{et al.} 
\cite{popl19}, which spells out a string constraint language with semantic 
conditions 
for decidability that directly use the notion of monadic decomposability of 
relations over strings. Loosely speaking, a constraint is simply a sequence of
program statements, each being either an assignment or a conditional:
\begin{align*}
    S ::= \qquad y := f(x_1,\ldots,x_r) \ |\
    \text{\ASSERT{$g(x_1,\ldots,x_r)$}}\ |\
            S; S\
           % \label{eq:symbex}
%\vspace{-1mm}
    %a ::= f(x_1,\ldots,x_n), \qquad b ::= g(x_1,\ldots,x_n)
\end{align*}
where $f: (\Sigma^*)^r \to \Sigma^*$ is a partial string function and 
$g\subseteq (\Sigma^*)^r$ is a string relation. The meaning of a constraint
is what one would expect in a program written in a standard imperative
programming language, which should support assignments and assertions. Note that
loops are not
allowed in the language since their target application is symbolic
executions (e.g.~see \cite{symbex-survey}). They provided two semantic 
conditions for ensuring decidability, one of which requires that each 
conditional $g$ is effectively monadic decomposable.
There is evidence (e.g.~\cite{Vijay-length,popl19}) that some form of
length reasoning in $g$ is indeed required for many applications of symbolic 
executions 
of string-manipulating programs, but much of the length constraints could be 
(not yet fully automatically) translated to regular constraints.
A potential application for our results is therefore to provide support for
complex string relations for $g$ in the form of regular relations, which permit
a rather expressive class of conditionals (e.g. some form of length reasoning, 
etc.). Despite this, this application also highlights what is currently missing
in the entire literature of monadic decomposability of rational relations:
a study of the problem of \emph{outputting} the monadic decompositions of
the relations, if monadic decomposable. (In fact, this is also true of other logical 
theories before the recent work of Veanes \emph{et al.} \cite{VBN+17}.)
What is the complexity of this problem with various representations of
recognizable relations (e.g. finite unions of products, boolean combinations of
regular constraints, etc.)? Although our results provide \emph{a first step} 
towards solving this function problem, we strongly believe this to be a highly challenging open problem in its own right that
deserves more attention.

%, i.e., there is an algorithm which provides a
%monadic decomposition for $g$. 

\bibliographystyle{plainurl}
\bibliography{references}

%\newpage
%\appendix
%
%\begin{center}
%    \huge \bfseries APPENDIX
%\end{center}
%
%\input{appendix}
%\input{example}

\end{document}